\documentclass[aps,pra,twocolumn,superscriptaddress,groupedaddress]{revtex4-1}
\usepackage{graphicx}  
\usepackage{dcolumn}   
\usepackage{bm}        
\usepackage{amssymb}   
\usepackage{physics}
\usepackage{tikz}
\usepackage{gensymb}
\usetikzlibrary{shapes,snakes}
\usetikzlibrary{shapes.geometric}
\usepackage{pgfplots}
\usepackage{chngcntr}
\usepackage{dsfont}
\usepackage[caption = false]{subfig}
\usepackage{hyperref}
\usepackage{xcolor}

\usepackage{float}
\usepackage[utf8]{inputenc}
\usepackage[english]{babel}
 \usepackage{amsthm}
\newtheorem{theorem}{Theorem}

\theoremstyle{definition}

\newcommand{\newc}{\newcommand}
\newc{\beq}{\begin{equation}}
\newc{\eeq}{\end{equation}}
\newc{\kt}{\rangle}
\newc{\br}{\langle}
\newc{\beqa}{\begin{eqnarray}}
\newc{\eeqa}{\end{eqnarray}}
\newc{\pr}{\prime}
\newc{\longra}{\longrightarrow}
\newc{\ot}{\otimes}
\newc{\rarrow}{\rightarrow}
\newc{\h}{\hat}
\newc{\bom}{\boldmath}
\newc{\btd}{\bigtriangledown}
\newc{\al}{\alpha}
\newc{\be}{\beta}
\newc{\ld}{\lambda}
\newc{\sg}{\sigma}
\newc{\p}{\psi}
\newc{\eps}{\epsilon}
\newc{\om}{\omega}
\newc{\mb}{\mbox}
\newc{\tm}{\times}
\newc{\ra}{\rightarrow}
\newc{\non}{\nonumber}
\newc{\ul}{\underline}
\newc{\hs}{\hspace}
\newc{\longla}{\longleftarrow}
\newc{\ts}{\textstyle}
\newc{\f}{\frac}
\newc{\df}{\dfrac}
\newc{\ovl}{\overline}
\newc{\bc}{\begin{center}}
\newc{\ec}{\end{center}}
\newc{\dg}{\dagger}
\newc{\T}{\mathcal{U}}
\newc{\Tp}{\mathcal{V}}
\newc{\J}{\mathsf{J}}
\newc{\sfL}{\mathsf{L}}
\newc{\C}{\mathsf{C}}
\newc{\B}{\mathsf{M}}
\newc{\V}{\mathsf{V}}
\newc{\tl}{\tilde}

\begin{document}

\widetext

\title{Quantum walk on a toral phase space}

\author{Sivaprasad Omanakuttan}
\email[]{sivaprasadto0811@gmail.com}
\author{Arul Lakshminarayan}
\email[]{arul@physics.iitm.ac.in}
\affiliation{Department of Physics, Indian Institute of Technology Madras, Chennai 600036, India}

\date{\today}

\begin{abstract}
A quantum walk on a toral phase space involving translations in position and its conjugate momentum is studied in the simple context of a coined walker in discrete time. The resultant walk, with a family of coins parametrized by an angle is such that its spectrum is exactly solvable with eigenangles for odd parity lattices being equally spaced, a feature that is  remarkably independent of the coin. The eigenvectors are naturally specified in terms the ~$q-$Pochhammer symbol, but can also be written in terms of elementary functions, and their entanglement can be analytically found. While the phase space walker shares many features in common with the well-studied case of a coined walker in discrete time and space, such as ballistic growth of the walker position, it also presents novel features such as exact periodicity, and formation of cat-states in phase-space. Participation ratio (PR) a measure of delocalization in walker space is studied in the context of both kinds of quantum walks; while the classical PR increases as $\sqrt{t}$ there is a time interval during which  the quantum walks display a power-law growth $\sim t^{0.825}$. Studying the evolution of coherent states in phase space under the walk enables us to identify an Ehrenfest time after which the coin-walker entanglement saturates.
\end{abstract}

\maketitle

\section{Introduction}

Quantum walks have been studied vigorously in the recent past and come in several flavors\citep{aharonov1993quantum,kempe2003quantum,
nayak2000quantum,childs2002example,
0305-4470-35-12-304,schmitz2009quantum}. Their potentially uses include quantum search algorithms \cite{aaronson2003quantum, childs2004spatial}
and universal quantum computation \citep{childs2009universal,lovett2010universal}. They come in continuous and discrete time versions and in many different settings \cite{Childs2010}\cite{bru2016quantum}.  In the simplest one, there is a two-dimensional coin space and a linear lattice space of discrete states in which the walker can jump by a step forward or backward. This quantizes the simplest model of a classical random walk which is recovered if the state of the coin is measured at every step \citep{kempe2003quantum,nayak2000quantum}. A quantum walk avoids such measurement and the peculiarities of the resulting dynamics is due to quantum interference between many possible paths. It is well-known that this results in  the walker's standard deviation increasing linearly in time ($\sim t$), in contrast to the case of the  diffusive classical walker ($\sim t^{1/2}$) \cite{kempe2003quantum}. 

One motivation for the present study is to introduce the non-commutative aspect, so central to quantum mechanics, in the walker dynamics. As the basic non- commutativity is between position and momentum, we formulate the walk as happening in ``phase space". In the quantization of classically chaotic systems, the non-commutative nature of conjugate variables has dramatic effects and effectively smoothens the classical mixing and can lead to dramatic quantum localization effects. For example a classical kicked rotor can spread diffusively in momentum and get completely localized quantum mechanically \cite{10.1007/BFb0021757,PhysRevLett.73.2974}. On the other hand quantum walks have an opposite tendency and spread faster than classical random walks. Of course the understanding of ``classical" in these two contexts are different, one being the standard dimensionless Planck constant tending to zero, while in the other, it is the effect of frequent projective measurements. 

Another motivation is that the study of these almost natural operators leads to solvable models with a very different spectral nature than that of the standard quantum walks hitherto examined. For example the translational invariance that leads to momentum conservation and makes the Fourier transform block-diagonalize the operator in the walk space
is now broken in the phase space walk. Yet, quite surprisingly, these models are still analytically tractable having equally spaced eigenangles for a continuous family of coin operators. They also lead to a linear increase in the standard deviation of both the position as well as the momentum.
In particular the phase space topology that we will study is the torus, that results from imposing periodic boundary conditions on both position and momentum.
The advantage of treating quantum algorithms on a torus has already been explored  
\citep{miquel2002quantum}, and it seems natural to also study quantum walks in this setting. 

The generic structure of a coin with $d$ states is comprised of a coin step and a walker step. The walker step is the unitary operator:
\beq
U_w=\sum_{i=0}^{d-1} A_i \otimes |i \kt \br i |,
\eeq
where $A_i$ are unitary operators on the walker space and the $\{|i\kt \}$ form
an orthonormal basis in the coin space. The coin dynamics is given by $\mathds{1}_w \otimes U_c$, where $U_c$ is some unitary operator on the coin space. The quantum walk is the combination
\beq
U_A=  U_w (\mathds{1}_w \otimes U_c),
\eeq
and the dynamics is simply powers $U_A^n$ of this operator. The canonical walk consists of the case when $d=2$ and $A_1=U$ and $A_2=U^{-1}$, where $U$ is the
translation, or shift, operator on the one-dimensional lattice. Thus the walk can be interpreted as being either ``forward" or ``backward" in position. The present work
explores the case when $d=2$ still but $A_1$ and $A_2$ do not commute. In particular it is natural to consider $A_2$ as being diagonal in the space of lattice states. In other words $A_1$ is a position translation operator while $A_2$ is a momentum translation operator, and in this sense the walk is a phase space walk. While  $A_i$ can be any unitary operator, those of greatest interest must be those that are local in some sense, so that there is a notion of continuous transport in some space. 
The case of generic $A_i$ could also have relevance and interest, but we do not consider them here. 

As far as we are aware this scenario has not been explored in the literature. The scenarios that have been studied and which may be related are of two kinds. One 
also calls it a quantum walk in phase space, but essentially the lattice is visualized as ``clock states" with equally spaced angles \cite{xue2008quantum}. Thus the walk is on a circle in phase space, with the phase space being action-angle. This is rather close to the walk on the line but for boundary conditions, in particular the two operators $A_1$ and $A_2$ still commute. The other is called an ``electric walk" in the literature and more closely allied to the present work. However there are two walker steps, in one of which $A_1$ and $A_2$ are the position translation operator and its inverse, while the other consists of the momentum translation and its inverse. Actually the set-up of the electric walk is more general than this, as the second step is interpreted as the effect of an electric field on the (charged) walker \citep{genske2013electric,di2014quantum,
bru2016electric}.  Another work where a similar structure has appeared previously concerns  quantum walks in non-Abelian discrete gauge theory   \cite{arnault2016quantum}. Thus while there are similar models that have been studied, we believe there is value to studying this particular variation, because of its simple interpretation and mathematical structure.

The plan of the paper is as follows: in section II, after reviewing briefly the discrete quantum walk in position space, we formulate the discrete version of the quantum walk on a toral phase space, and also discuss classical aspects of the walk. The  spectra of the phase space quantum walk is then solved for exactly. It is found that the eigenangles are rational multiples of $\pi/N$ if there are odd number ($N$) of lattice sites, and hence there exists a time $(2 N)$ when the walker dynamics becomes identity. The eigenvectors can be written in terms of the somewhat esoteric $q-$Pochhmammer symbols  which are fundamental to the theory of $q-$series \cite{koekoek1996askey,gasper} used in generalized Hypergeometric functions and combinatorics. It maybe remarked that the usual conifguration space walk with periodic boundary conditions (referred to below as CSW, as opposed to the phase space walk, which is referred to as PSW) has a more complex spectrum despite having translational symmetry and that it does not have the periodicity that is observed in the phase space walk \cite{nayak2000quantum}. The phase space walk is studied for a whole family of coins parametrized by an angle, and certain remarkable results of the spectra are found. Notably the eigenangles do not change across the family, while all the complexity of the walk is encoded in the eigenstates, Also the entanglement in the eigenstates is exactly computed.
 
In section III, we  discuss the evolution of states using the walk, and study both position eigenstates as well as coherent states. One measure that we study in addition to the standard deviation is the participation ratio of the walker. This measures the ``delocalization" of the walker across the lattice space. This can distinguish between ballistic spreading with linear growth of standard deviation from a non-trivial walk. We provide evidence that both the usual quantum walk on the line and for the walk in phase space the participation ratio increases as $\sim t^{0.825}$, while the classical walk participation ratio increases slower as $\sim \sqrt{t}$. It is shown that the participation ratio and standard deviation share features with the usual CSW. In particular it is seen here that the quantum participation ratio increases at a rate that is larger than the classical. The evolution of coherent states gives rise to ``cat states" in phase space and is akin to similar results found recently for the Gaussian states in one-dimensional configuration space walks  \cite{zhang2016creating}. The entanglement between walker and coin is also found and the phase space cat-states are essentially formed when the entanglement saturates. We end with a summary and discussions in section IV.

\section{Discrete quantum walk in phase space: Definition and spectra}
The standard quantum walk with an orthogonal matrix for the coin operator in configuration space (CSW) can be described in terms of the lattice (``position" eigenkets) states $\{ |n\kt \}$ as 
\begin{widetext}
\beq
\begin{aligned}
\label{eq:csw}
U_{csw}(\theta)= \left(|0\kt \br 0 |\otimes \sum_{n} |n+1 \kt \br n|+|1 \kt \br 1 |\otimes \sum_{n} |n \kt \br n+1| \right) \left(U_{\theta\,} \otimes \mathds{1}\right),
\end{aligned}
\eeq
\end{widetext}
where $U_{\theta\,}$ is given as 
\beq
\label{eq:coin}
U_{\theta\,}= \left( \begin{array}{rr} \cos\theta\, & \sin\theta\, \\ \sin\theta\,& -\cos\theta\, \end{array} \right).
\eeq
The walk considered in this paper is a simple modification:
\begin{widetext}
\beq
\label{eq:psw}
U_{psw}(\theta)= \left(|0\kt \br 0 |\otimes \sum_{n} |n+1 \kt \br n|+|1 \kt \br 1 |\otimes \sum_{n} e^{i \alpha n} \,|n \kt \br n| \right) \left(U_{\theta\,} \otimes \mathds{1}\right),
\eeq
\end{widetext}
where $\alpha$ is a real constant, and is referred to as a phase space walk. The lack of translational invariance is explicit and there is no conserved quasi-momentum. This is structurally similar to tight-binding Hamiltonians with on-site potential, in particular the Harper Hamiltonian with a potential that is $\cos(q)$. The case when $e^{i \alpha}$ is 
a root of unity, say $\omega=e^{2 \pi i /N}$ where $N$ is an integer is considered below. Other possibilities are of interest, however we wish to study this in a setting of a walk as on a phase space lattice of size $N \times N$. In other words the phase space walk in Eq.~(\ref{eq:psw}) is on a phase space torus and is on a finite $N-$ dimensional Hilbert space \cite{schwinger1960unitary}.

The position translation operator $\T$  acting on the position eigenkets shifts them:
\begin{equation}
\label{eq:position space1}
\T\vert{n}\rangle=\vert{n+1\pmod N}\rangle,
\end{equation} 
where $0\leq n<N$ and $N$ is the total number of lattice sites and we use periodic boundary condition, $\T^N=\mathds{1}$.
The momentum states $|\tl{k}\kt$ are eigenvectors of $\T$:
\begin{equation}
\label{eq:position space2}
  \T |\tl{k}\kt =\omega^{-k}|\tl{k}\kt,
 \end{equation} 
The momentum translation operator $\Tp$ is such that if $l=k+1 \, \text{mod}\, N$,
\begin{equation}
\label{eq:momentum space1}
\Tp |\tl{k}\kt=|\tl{l} \kt,\;\; \Tp|{n}\kt=\omega^{n}|{n}\kt.
\end{equation}
Their commutation relation is given by the Weyl relation
\begin{equation}
\label{eq:commutation relation}
\Tp\,  \T=\omega \, \T  \,\Tp.
 \end{equation}
 
The discrete Fourier transform which interchanges the role of position and momentum translation operator for the  periodic boundary condition  is given as \cite{schwinger1960unitary,saraceno1990classical} 
 \begin{equation}
 \label{eq: quantum fourier transform1}
\br n | \tl{k} \kt \equiv (G_N)_{nk}=\frac{1}{\sqrt{N}}\exp\left(2 \pi i kn/N\right).
  \end{equation}
The following transformation equations 
  \begin{equation}
  \label{eq: quantum fourier transform2}
  G_N \T G_N^{\dagger}=\Tp,\hspace{.3cm}G_N \Tp G_N^{\dagger}=\T^{\dagger},
  \end{equation}  
are readily verified.  
In terms of these operators the configuration space walk is $U_{csw}=(|{0}\kt \br{0}|\otimes \T +|{1}\kt\br{1}|\otimes \T^{\dagger})(U_{\theta\,}\otimes \mathds{1}_N)$ \cite{nayak2000quantum,kempe2003quantum,
PhysRevA.77.032326}. The quantum walk in phase space, and object of the present study is  
 \begin{equation}
 \label{eq:psw1}
 U_{psw}=\left [|{0}\kt \br{0}|\otimes \T +|{1}\kt\br{1}|\otimes \Tp\right](U_{\theta\,}\otimes \mathds{1}_N)
\end{equation} 
$U_{psw}$ can be written in  a block matrix form as, 
 \begin{equation}
  \label{eq:psw3}
 U_{psw}(\theta)=\left( \begin{array}{cc}
\cos\theta\, \T &\sin \theta\, \T \\
\sin\theta\,\Tp&-\cos\theta\,\Tp
\end{array} \right).
 \end{equation}
It maybe noted that the operators $\T$ and $\Tp$ are also referred to in the literature as a higher dimensional generalization of the Pauli matrices or ``clock" and ``shift" matrices and often denoted as $X$ and $Z$ \citep{0034-4885-67-3-R03,hegde2015unextendible} and the  classical limit of the quantum walk in phase space is described here along with its similarities to the classical random walk\citep{pearson1905problem,
chandrasekhar1943stochastic}.

 The scenario is very similar to the classical random walk such that depending on the  outcome of the coin the walker's next step is decided. The walker gets a boost of one unit in momentum or shifts in position depending on whether the outcome of the coin is head or tail respectively.  Let  $(p,q)$ be the representation of the walker in phase space where $p$ represents the momentum and $q-$ the position. Let the walker start from the origin $(0,0)$. After one time step  the walker will be either  at $(1,0)$ or $(0,1)$ with a probability of $\frac{1}{2}$ and after a two time steps  the walker can be at $(2,0)$ and $(0,2)$ with a probability of $\frac{1}{4}$ and at $(1,1)$ with a probability $\frac{1}{2}$. The above mentioned situation is depicted  in  Fig.~\ref{fig:Classical Random walk in Phase Space}. The   classical limit of quantum walk in phase space also follows a binomial distribution and tends in a standard manner to a Gaussian distribution as $t\to \infty$.

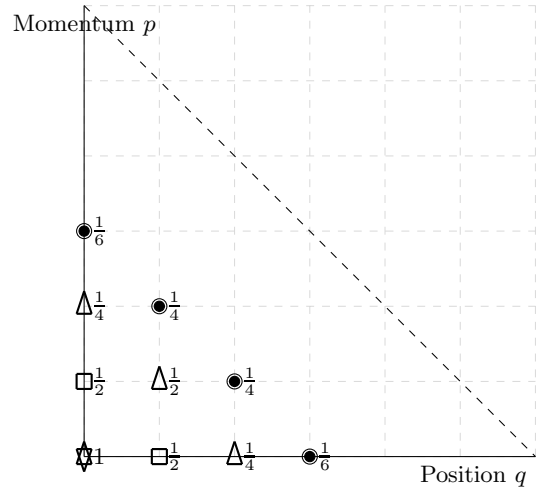
\begin{figure}[]

\centering
 \begin{tikzpicture}[square/.style={regular polygon,regular polygon sides=4}]
\draw[help lines, color=gray!30, dashed] (0,0) grid (6,6);
\draw[] (0,0)--(6,0) node[anchor=north  east] {Position $q$};
\draw[] (0,0)--(0,6)node[anchor=  north] {Momentum   $p$};

 
\node at (1.2,0){$\frac{1}{2}$};
\node at (0.2,1){$\frac{1}{2}$};
\node at (0.2,0){$1$};

 \node at (0.2,2){$\frac{1}{4}$};
  \node at (1.2,1){$\frac{1}{2}$};
        
 \node at (3.2,0){$\frac{1}{6}$};
  \node at (1.2,2){$\frac{1}{4}$};      
 \node at (0.2,3){$\frac{1}{6}$};
  \node at (2.2,1){$\frac{1}{4}$};
 \node at (2.2,0){$\frac{1}{4}$};       
    \draw[thick,-] (.9,.1)--(1.1,.1);
    \draw[thick,-] (1.1,.1)--(1.1,-.1);
    \draw[thick,-] (1.1,-.1)--(.9,-.1);
    \draw[thick,-] (.9,-.1)--(.9,.1);
     \draw[thick,-] (.9,.9)--(1.1,.9);
    \draw[thick,-] (1.1,.9)--(1.,1.2);
    \draw[thick,-] (1.,1.2)--(.9,.9);
     \draw[thick,-] (1.9,-.1)--(2.1,-.1);
    \draw[thick,-] (2.1,-.1)--(2.,.2);
    \draw[thick,-] (2.,.2)--(1.9,-.1);
     \draw[thick,-] (-.1,1.9)--(.1,1.9);
    \draw[thick,-] (.1,1.9)--(0,2.2);
    \draw[thick,-] (0.,2.2)--(-.1,1.9);

  \draw[thick,-] (.1,.9)--(.1,1.1);
    \draw[thick,-] (.1,1.1)--(-.1,1.1);
    \draw[thick,-] (-.1,1.1)--(-.1,.9);
    \draw[thick,-] (-.1,.9)--(.1,.9);
    
    \draw (3,0) circle (.1cm);
    \draw[black,fill=black] (3,0) circle (.5ex);
      \draw (0,3) circle (.1cm);
    \draw[black,fill=black] (0,3) circle (.5ex);
    \draw (2,1) circle (.1cm);
    \draw[black,fill=black] (2,1) circle (.5ex);
      \draw (1,2) circle (.1cm);
    \draw[black,fill=black] (1,2) circle (.5ex);
     \draw[thick,-] (0,0.2)--(-.1,-.1);
 \draw[thick,-] (-.1,-.1)--(.1,-.1);
 \draw[thick,-] (.1,-.1)--(0,0.2);
\draw[dashed](0,6)--(6,0);
 \draw[thick,-] (0,-0.2)--(.1,.1);
  \draw[thick,-] (.1,.1)--(-.1,.1);
   \draw[thick,-] (-.1,.1)--(0,-0.2);
\end{tikzpicture}
\caption{The classical random walk in phase space is shown with the corresponding occupation probabilities for times $0,1,2\,\text{and}\,3$ (star, squares, rectangles and circles respectively), whereas the dashed line shows the front on which the walker moves in phase space.}\label{fig:Classical Random walk in Phase Space}
\end{figure}

If the probability of a shift in position is $f$, and a shift in momentum is $1-f$ ($0 \leq f  \leq 1$), the probability of the walker being at $(p,q)$ after time $t$ is  
\begin{equation}
\label{eq:classical random walk in phase space}
G(q,p,t)={t \choose q} f^q (1-f)^{(t-q)}\delta(p+q-t).
\end{equation}  
Hence the probability distribution for the  random walk in phase space  is similar to  the usual random walk except that this happens along the line $p+q=t$. The quantum analog of the random walk in phase space  is  discussed in detail in the next section.
\subsection{Spectra of the phase space walk}
The stationary state properties, the eigenvalues and eigenvectors, are naturally of interest and enable solutions of time evolution problems as well. Here we show
that the spectra of the phase space walk can be analytically found.
\begin{theorem}
The eigenvector corresponding to eigenvalue $\lambda_k$ for $U_{psw}(\theta)$, for $\theta\neq 0$,  on $N$ lattice sites in the position basis $(|{n})\kt$ is given by 
\begin{equation}
\label{eq:eigenvector1}
|{\phi_k}\kt= \frac{1}{\sqrt{C_N(k)}} \sum_{n=0}^{N-1} \left(a_n(k)|{n}\kt|{0}\kt+b_n(k)|{n}\kt|{1}\kt \right)
\end{equation}
where $C_N(k)$ is a normalization constant, and $a_n(k)$, $b_n(k)$ are given by   ~$q-$Pochhammer symbols:  
\begin{subequations}
\label{eq:eigenvector2}
\begin{align}
a_n(k)=&\omega^{-\frac{n(n-1)}{2}}\frac{\left(-\sec\theta\,\, \lambda_k^{-1};\omega\right)_n}{\left(-\sec\theta\, \,\lambda_k;\omega^{-1}\right)_{n}},\label{subeqn:an}\\\
b_n(k)=&\omega^{-\frac{n(n-1)}{2}} \tan\theta\,\frac{\left(-\sec\theta\,\,\lambda_k^{-1};\omega\right)_n}{\left(-\sec\theta\,\, \lambda_k;\omega^{-1}\right)_{n+1}}\label{subeqn:bn}.
\end{align}
\end{subequations}
The eigenvalues $\lambda_k = \omega^{k/2}, \; k=0,1,\cdots, 2N-1, \; \text{if } N \text{is odd}$. If $N$ is even, they come in ``split pairs": $\lambda^{\pm}_{k}=
\exp (\pm i \alpha) \;\omega^k, \;k=0, 1, \cdots, N-1$, and 
where  $\alpha=\frac{1}{N}\cos^{-1}(\cos^N(\theta\,)), 0\leq \alpha N \leq \frac{\pi}{2}$. The normalization is determined by $C_N(k)=2N/(1+(-1)^k \cos^N\theta\,)$ if $N$ is odd and simply $=2N$ if $N$ is even.
\end{theorem}
\begin{proof}
The eigenvalue equation $ U_{psw} (\theta)|{\phi_k}\kt=\lambda_k |{\phi_k}\kt$ implies the following recursion relations between the coefficients $a_n(k)$ and $b_n(k)$
\begin{subequations}
\label{eq:iteration 1}
\begin{align}
\sin\theta\,\, a_n(k)-\cos\theta\,\, b_n(k)=\lambda_k b_n(k)\omega^{-n},\label{subeqn:ini1}\\ \cos\theta\, \, a_{n-1}(k)+\sin\theta\, \, b_{n-1}(k)= \lambda_k a_{n}(k)\label{subeqn:ini2},
\end{align}
\end{subequations} 
which are valid for all $0 \le n \le N-1$ and $a_{-1}(k)=a_{N-1}(k)$ while $b_{-1}(k)=b_{N-1}(k)$.
Using the two recursions, simple algebra that eliminates the $b$ variables, yields
\beq
a_n(k)=a_{n-1}(k) \left( \dfrac{\lambda_k^{-1} + \cos \theta\, \, \omega^{-n+1}}{\cos \theta\, + \lambda_k \, \omega^{-n+1}}\right).
\eeq
With $a_0(k)=1$, which can be assumed (provided it is non-zero) as we are going to fix the normalization later, the expression for $a_n(k)$ is given as 
\begin{equation}
\label{eq:iteration 2}
a_n(k)=\omega^{-\frac{n(n-1)}{2}}\prod_{j=0}^{n-1}\left(\frac{ 1+ \sec\theta\, \,\lambda_k^{-1} \omega^{j} }{1+\sec\theta\,\, \lambda_k \, \omega^{-j}} \right).
\end{equation}
It then follows from the fact that the $\lambda_k$ are on the unit circle, that  $|a_n(k)|=1$, and hence all the $a_n(k)$ are pure phases. Using  the ~$q-$Pochhammer symbol  defined  as:
 \begin{equation}
 \label{eq:q-Pochhammer}
 \begin{aligned}
 (x;q)_n=\prod_{j=0}^{n-1}(1-x \,q^j),
 \end{aligned}
 \end{equation}
 and the relation 
 \beq
 \label{eq:b-a:relation}
b_n(k) =\frac{a_n(k) \tan \theta\,}{1+\sec \theta\, \, \lambda_k\, \omega^{-n}}
\eeq 
 obtained from Eq.~\eqref{eq:iteration 1} gives $b_n(k)$, and hence the expressions for the eigenvectors in Eq.~(\ref{eq:eigenvector2}) follow. However these contain the as yet undetermined eigenvalues which we now turn to.
 
The eigenvalues can be calculated using the equation  $\cos\theta\,\, a_{N-1}(k)+\sin \theta\, \, b_{N-1}(k)=\lambda_k \, a_0(k)=\lambda_k$, which follows on putting $n=0$ in the  Eq.~\eqref{subeqn:ini2}. Using Eq.~(\ref{eq:b-a:relation}) this leads to
\beq
\lambda_k = a_{N-1}(k) \left( \dfrac{\sec \theta\, + \lambda_k \,  \omega^{-(N-1)}}{1+\sec \theta\, \, \lambda_k \,\omega^{-(N-1)}}\right).
\eeq
The eigenvector component $a_{N-1}(k)$ is now written in terms of the eigenvalue from Eq.~(\ref{subeqn:ini2}) and some algebra which takes into account that $\omega^{-N(N-1)/2}=(-1)^{N-1}$, leads finally to an equation containing only the eigenvalue that is sought: 
\beq
\dfrac{(- \sec \theta\, \, \lambda_k^{-1};\omega)_{N}}{(- \sec\theta\, \, \lambda_k;\omega^{-1})_{N}}=(-1)^{N-1}.
\eeq
Noting a ~$q-$Pochhammer  identity: $\left(x;\exp(2\pi i/n) \right)_n=1-x^n,$ gives
\beq
\label{eq:eigenvalue-q Pochammer}
1-(-\sec\theta\,\, \lambda_k^{-1})^N=(-1)^{N-1}\left(1-(-\sec\theta\,\, \lambda_k)^N\right).
\eeq
For odd values of $N$ the above equation simplifies to $\lambda_k^{2N}=1$, all dependence on $\theta\,$ remarkably disappearing. 
One can argue that this implies that eigenvalues are {\it  all} the $2N$-th roots of unity.
This follows from Eq.~(\ref{eq:eigenvector2}), each such root giving a different eigenvector of $U_{psw}(\theta)$, 
and as all the eigenvectors of a unitary operator are necessarily orthogonal, this forms a complete basis. 
Hence the eigenvalues for $N$ odd are  
\beq
\label{eq:eigenvalue odd}
\lambda_{k}=\omega^{k/2}, \;\;0\leq k\leq 2N-1.
\eeq 
Note that for odd values of $N$, one could write these as $\lambda_{l}^{\pm}=\pm \omega^l$, with $0 \leq l <N$. These coincide with the eigenvalues of $\T$ and $-\Tp$, which are the blocks that $U_{psw}(0)$ contains. Hence the family of operators $U_{psw}(\theta)$ have the same eigenvalues but different eigenvectors.

For  even values of $N$, letting $\lambda=\exp(i\beta)$, Eq.~\eqref{eq:eigenvalue-q Pochammer} implies that $\exp(i\beta N)+\exp(-i\beta N)=2\cos^N(\theta\,)$ and hence $\beta=\pm\alpha+2\pi l/N$ where $\alpha=\frac{1}{N}\cos^{-1}(\cos^N \theta\,)$ and $l$ is any integer. 
\beq
\lambda^{\pm}_l=\exp(\pm i \alpha)\omega^l, \;\;0\leq l < N.
\label{eq:gh}
\eeq
Thus the eigenvalues of $U_{psw}(\theta)$ change in this case with $\theta$, they start out doubly degenerate at $\theta=0$, become non-degenerate for $\theta>0$, and end up being equally spaced on the unit circle at $\theta=\pi/2$. The normalization constants follows from the $a_{n}$ and $b_n$, details are relegated to Appendix~(\ref{app:norm}). 
\end{proof}

Thus the spectrum of the walk $U_{psw}(\theta)$ presents interesting mathematical structures. The eigenvectors of the PSW have been written in terms of the $q-$Pochhammer symbol in Eq.~\eqref{eq:eigenvector2}. However they can be written simply in terms of trigonometric functions. For an eigenstate in a lattice with odd number of sites and eigenvalue $\omega^{k/2}$, 
\beq
\label{eq:eqa 1}
a_n(k)=\exp(-i \zeta_n(k)),
\eeq
 where for $n>0$ 
 \beq
 \label{eq:lll}
 \zeta_n(k)=\sum_{j=0}^{n-1}\left(\frac{2 \pi}{N}j-2\tan^{-1}\dfrac{\sin \frac{2\pi}{N}(j-\frac{k}{2})}{\cos \theta +\cos \frac{2\pi}{N}(j-\frac{k}{2})}\right),
 \eeq
 and $\zeta_0(k)=0$. The $b_n(k)$ can then be found using the relation Eq.~(\ref{eq:b-a:relation}). 
 
 The marginal cases $\theta=0$ and $\theta=\pi/2$ give more insight into the spectra of the walk and therefore we turn to these.
  \subsubsection{Case I: $\theta =0$}
 For $\theta=0$ the block matrix representation of $U_{psw}$ in Eq.~\eqref{eq:psw3} simplifies to a block diagonal form given as
\begin{equation}
  U_{psw}(0)=\left( \begin{array}{cc}
 \T & 0 \\
0&-\Tp
\end{array} \right).
 \end{equation} 
Hence the eigenvectors $|{\phi_k}\kt$ are direct products of the eigenvectors of $\T$ or $\Tp$ with $|{0}\kt$ or $|{1}\kt$ respectively and from their definitions these are momentum and site/position eigenstates respectively. For simplicity we discuss the odd $N$ case further. For the eigenvalue given by $\omega^{k/2}$, the corresponding state is 
\begin{equation}
\label{eq:eigenvec}
|{\phi_k}\kt=\begin{cases}
  \frac{1}{\sqrt{N}}\sum_{n=0}^{N-1}\exp(-\frac{2 \pi i n k }{2N})|{n}\kt|{0}\kt,& \text{if } k \,\,\text {is even}\\
           |{\frac{1}{2} (k-N)}\kt|{1}\kt    &\text{for}\,k \,\, \text{ odd.}
    \end{cases}
\end{equation}
The above is consistent with the expression in Eq.~(\ref{eq:lll}) when $\theta=0$. The eigenvectors 
of $\T$ are such that the $a_0 \neq 0$, and the analysis holds. Also the corresponding $b_n$ all vanish. This half of the spectrum is when $k$ is even (say $k=2l$, with $0 \leq l \le N-1$) the eigenvalues are simply $\omega^l$ and the eigenstates are momentum states in the walker space localized in momentum to $l=N-k/2$. The other half of the states, for $k$ odd, the component $a_0$ can vanish and needs to be treated specially. Going back to the basic equation in Eq.~(\ref{subeqn:ini1}) it follows that for $\theta=0$, either $b_n(k)=0$ or $\lambda_k \omega^{-n}=-1$, which implies that the state labelled by $k$ is localized at position $n=((k-N)/2)\, \text{mod}\, N$. 

Thus position delocalized (momentum localized) and localized states alternate on the eigenangle circle.The entanglement of the eigenvectors are zero for the case in which $\theta=0$, but become non-zero for $\theta>0$, when the states contain both these delocalized and localized parts. The delocalized part is associated with finding the walker state when the measurement of the coin results in state $|0 \kt$, while the walker is by and large localized when the measurement results in the state $|1\kt$. The localized part also gets increasingly delocalized as $\theta$ increases till $\theta=\pi/2$ when both parts are maximally delocalized, being pure phases. Some eigenvectors corresponding to $\theta\,=\pi/4$ are shown in the Fig.~\ref{fig:modulus}, where the phase of $a_n$ and $b_n$ are shown along with the magnitude of $b_n$, which is  the localized component of the state.

\begin{figure*}[]
\subfloat{\includegraphics[width =2.2in]{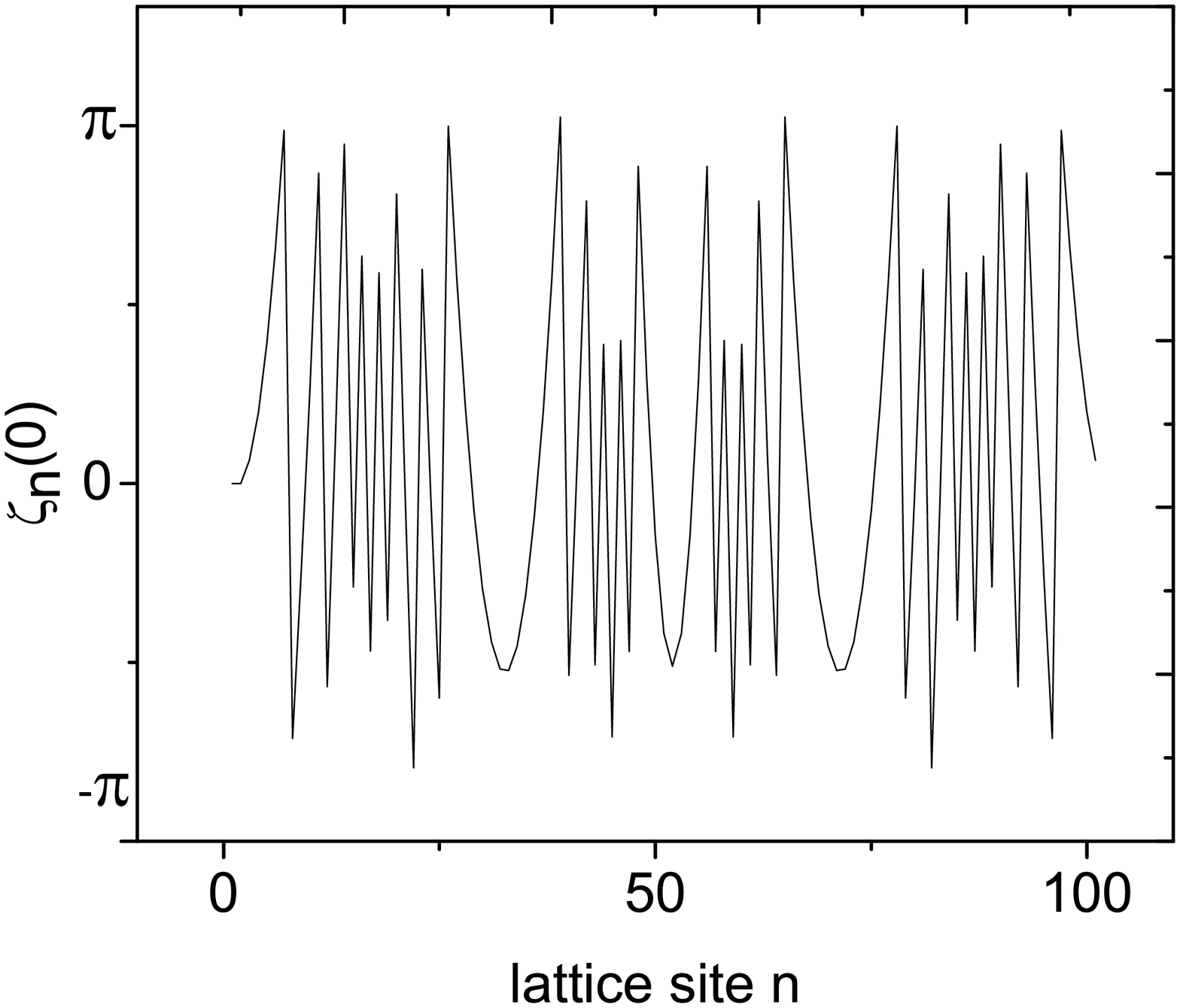}}
   \subfloat{\includegraphics[width =2.2in]{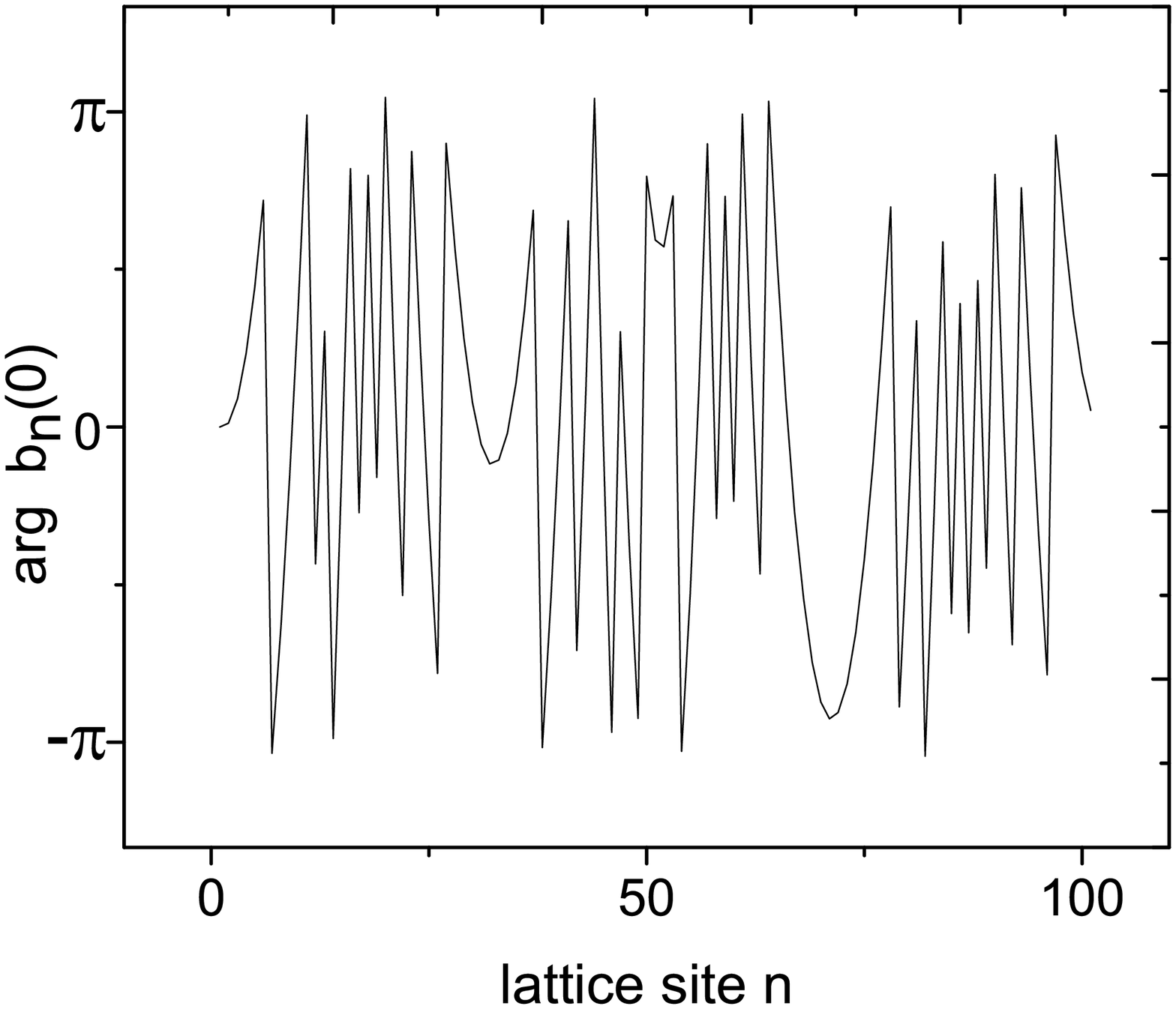}}
\subfloat{\includegraphics[width =2.2in]{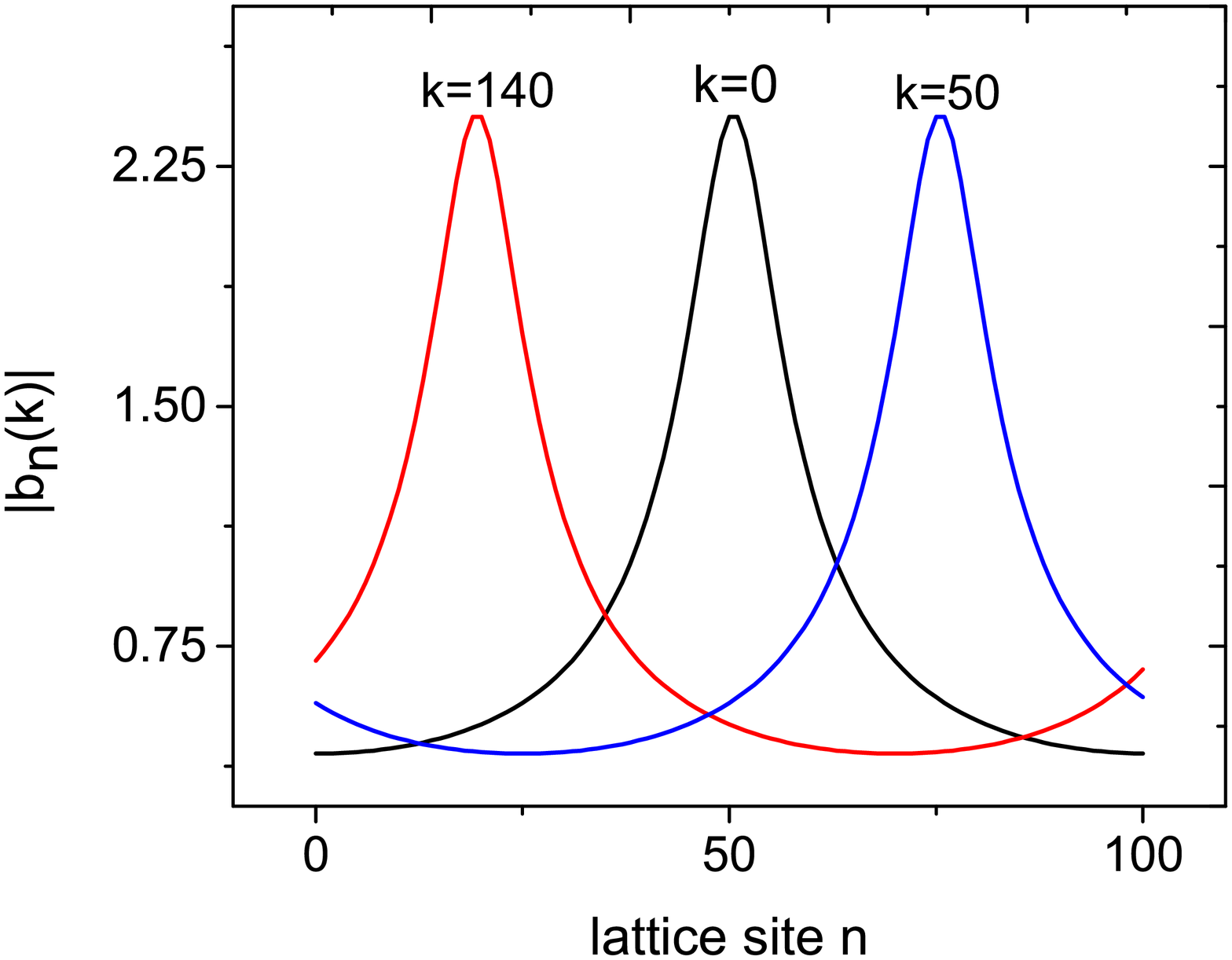}}
\\
\caption{The phases of  the eigenvector components $a_n(0)=\br n| \br 0| \phi_0\kt$ (left), and $b_n(0)=\br n| \br 1| \phi_0\kt$ (middle) are shown  while in the right is shown $|b_n(0)|$, along with $|b_n(50)|$ and $|b_n(140)|$ as an illustration of the structure of eigenstate of the PSW for the case $N=101$ and $\theta=\pi/4$.}\label{fig:modulus}
\end{figure*}

 \subsubsection{Case II: $\theta =\pi/2$}
 
 In this case using Eq.~(\ref{eq:lll}) 
 gives,
 \beq
 a_n(k)=\exp\left(\frac{ \pi i}{N} [n(n-1)-2 n k] \right),
 \eeq
and from Eq.~\eqref{eq:b-a:relation} one gets the other half: 
 \beq
 b_n(k)=\exp\left (\frac{  \pi i }{N}[n(n+1)-2 k n-k ] \right).
 \eeq
 Hence eigenvectors for the case of odd $N$ and $\theta=\pi/2$ are 
 \beq
 |{\phi_k}\kt=\frac{1}{\sqrt{2N}}\sum_{n=0}^{N-1} (\omega^{n^2-n-2 nk})^\frac{1}{2}\left[ |n\kt |0 \kt +\omega^{n-k/2} |n \kt |1\kt \right].
 \eeq
Thus in this case all the components are pure phases and the eigenstates are completely delocalized in both site and momentum space (see further below for discussions related to momentum space), and have maximum coin-walker entanglement. In fact it is interesting to calculate the evolution of entanglement with $\theta$. 

\subsubsection{Entanglement of Eigenvectors}
We concentrate again on the $N$ odd case for simplicity. From Eq.~\eqref{eq:eigenvector1}, the reduced density matrix of the 
coin in state $|\phi_k\kt$ is given as 
 \beq
\rho_{k}=\frac{1}{C_N(k)} \left( \begin{array}{ll} N &\sum_{n=0}^{N-1}a_n(k) b_n^*(k) \\\sum_{n=0}^{N-1}a_n^*(k) b_n(k)& C_N(k)-N \end{array} \right).
\eeq
Now using the Eq.~\eqref{eq:b-a:relation} yields,
\beq
\begin{aligned}
\label{eq:at1}
\sum_{n=0}^{N-1}a_n(k) b_n^*(k) =&\sin \theta \sum_{n=0}^{N-1}\frac{1}{\cos \theta+ \lambda_k \omega^{-n}}\\
=&(-1)^k \frac{N \tan \theta \cos^N \theta}{1+(-1)^k \cos^N \theta}.
\end{aligned}
\eeq
A proof of the sum appearing here is given in Appendix~(\ref{app:off diagonal}). 
Using the value of the normalization $C_N(k)$ given in Eq.~\eqref{eq:nn1} the reduced density of the coin is,
\beq
\rho_{k}=\frac{1}{{2}} \left( \begin{array}{rr} 1+(-1)^k \cos^{N}\theta & (-1)^k \cos^{N}\theta \tan \theta  \\(-1)^k \cos^{N}\theta \tan \theta& 1-(-1)^{k} \cos^{N}\theta \end{array} \right).
\eeq
The eigenvalues of the reduced density matrix are,
\beq
\mu_{k}^{\pm}=\frac{1}{2}( 1\pm \cos ^{N-1}\theta),
\eeq
with no dependence on the state index $k$, which implies that all eigenstates are uniformly entangled, indicating that local unitary operators may connect the eigenstates. Moreover for large enough $N$ and for $\theta$ not very close to $0$, these eigenvalues are nearly $1/2$ each and therefore all the eigenstates are nearly maximally entangled.

The von Neumann entropy  $(S_{vN})$ is the entanglement between coin and walker and is therefore
 \beq
 \begin{aligned}
 S_{vN}=&-\frac{1}{2}\left[( 1+ \cos ^{N-1}\theta)\ln(1+ \cos ^{N-1}\theta)\right]\\
 &-\frac{1}{2}\left[ ( 1-\cos ^{N-1}\theta)\ln(1- \cos ^{N-1}\theta) \right].
 \end{aligned}
 \eeq
 When $|\cos^{N-1}(\theta)| \ll 1$ this is approximately
 \beq
 S_{vN} \approx \ln(2)-\frac{1}{2}\cos ^{2N-2}\theta.
 \eeq
 The linear entropy is another widely used measure and for binary entropies it is monotonic with the von Neumann. 
 Define as $S_{l}=1-\Tr(\rho^2_k)$, this has a more explicit evaluation as
 \beq
 S_{l}=\frac{1}{2}(1-\cos ^{2N-2}\theta).
 \eeq
It is clear the eigenvectors go from being unentangled at $\theta=0$ to being maximally entangled at $\theta=\pi/2$. For large enough $N$, the increase is rapid, for example at $\theta=\pi/4$ corresponding to the Hadamard coin, the eigenstates linear entropy is uniformly $S_l=(1-2^{-N+1})/2$.

The PSW has a chiral symmetry \citep{asboth2016short} which implies the presence of pairs of eigenstates.
The  parity or reflection operator($R_N$) is $
 \br{n}|R_N|{n'}\kt=\delta[(N-n-n') \, \text{mod}\, N,0]$ in position (and also momentum) representations. This satisfies $R_N^{\dagger}=R_N$ and $
R_N\T R_N=\T^{\dagger},\,R_N \Tp R_N=\Tp^{\dagger}.
$
Defining another unitary operator  $\tilde{R}_{2N}$ as, $\tilde{R}_{2N}= U_{\theta}\otimes R_N$, where $U_{\theta}$ is the coin operator from Eq.~(\ref{eq:coin}), it follows that,
 \begin{equation}
\tilde{R}_{2N}U_{psw}\tilde{R}_{2N}=\left( \begin{array}{cc}
\cos\theta\,\T^{\dagger}& \sin\theta\,\Tp^{\dagger} \\
\sin\theta\,\T^{\dagger}&-\cos\theta\,\Tp^{\dagger}
\end{array} \right)=U_{psw}^{\dagger}.
\end{equation}
This implies that if $\lambda_k$ is an eigenvalue, so is $\lambda_k^*$,
and the corresponding eigenvectors are $|\phi_k\kt$ and $\tilde{R}_{2N}|\phi_k \kt$ respectively.

\section{Evolution of states under the phase space Hadamard walk}
This section is devoted to dynamical aspects of the walk and, unless otherwise stated, the coin is the Hadamard operator $H=U_{\pi/4}$. 
The probability distribution $p_n(t)$ for the phase space walk can be calculated using the stationary state properties discussed in the previous section. 

\subsection{Walker localized in position}

If the walker starts at the site origin $|n=0\kt$, with the coin-state $|{0}\kt$, the state of the walker after time $t$ is $U_{psw}^t(\pi/4)\ket{0}|{0}\kt$. Using the eigenvalue decomposition of the unitary operator the probabilities of finding the walker at a lattice site $n$ with coin-state $|{0}\kt$ or $|{1}\kt$ are,
\begin{equation}
\label{eq:prob5}
\begin{aligned}
p_n(t;0)=\left\vert \sum_{k=0}^{2N-1} \frac{1}{C_N(k)}\lambda_k^t a_n(k) \right\rvert^2, \\ 
p_n(t;1)=\left\vert \sum_{k=0}^{2N-1} \frac{1}{C_N(k)} \lambda_k^t b_n(k) \right\rvert^2.
\end{aligned}
\end{equation}
Here $a_n(k),b_n(k)$ are found from Eq.~\eqref{eq:eigenvector2} with $\theta=\pi/4$ and we have used $a_0(k)=1$. The probability of finding the walker at site $n$ after time $t$ is $p_n(t)=p_n(t;0)+p_n(t;1)$.

 All properties of the distribution can be calculated from the above equation efficiently, and there is no need for matrix diagonalization or powers. The plot of the probability distribution versus lattice sites using the above expression is given in Fig.~\ref{fig:Probability distribution} These bear a striking resemblance to the walker probability distributions for the case of the usual walk (CSW)\citep{kempe2003quantum}. 
\begin{figure*}
\subfloat{\includegraphics[width = 3in]{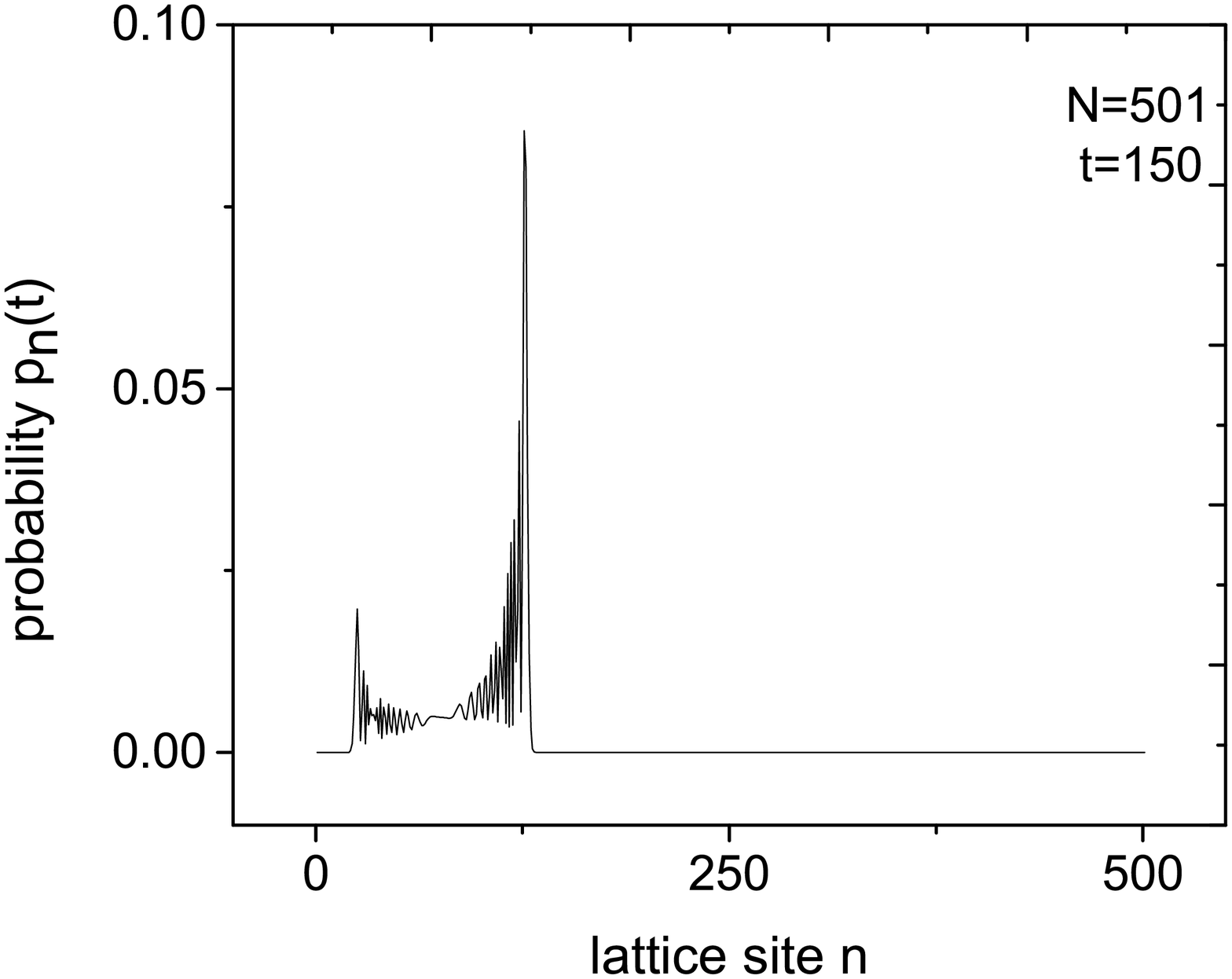}}
\subfloat{\includegraphics[width = 3in]{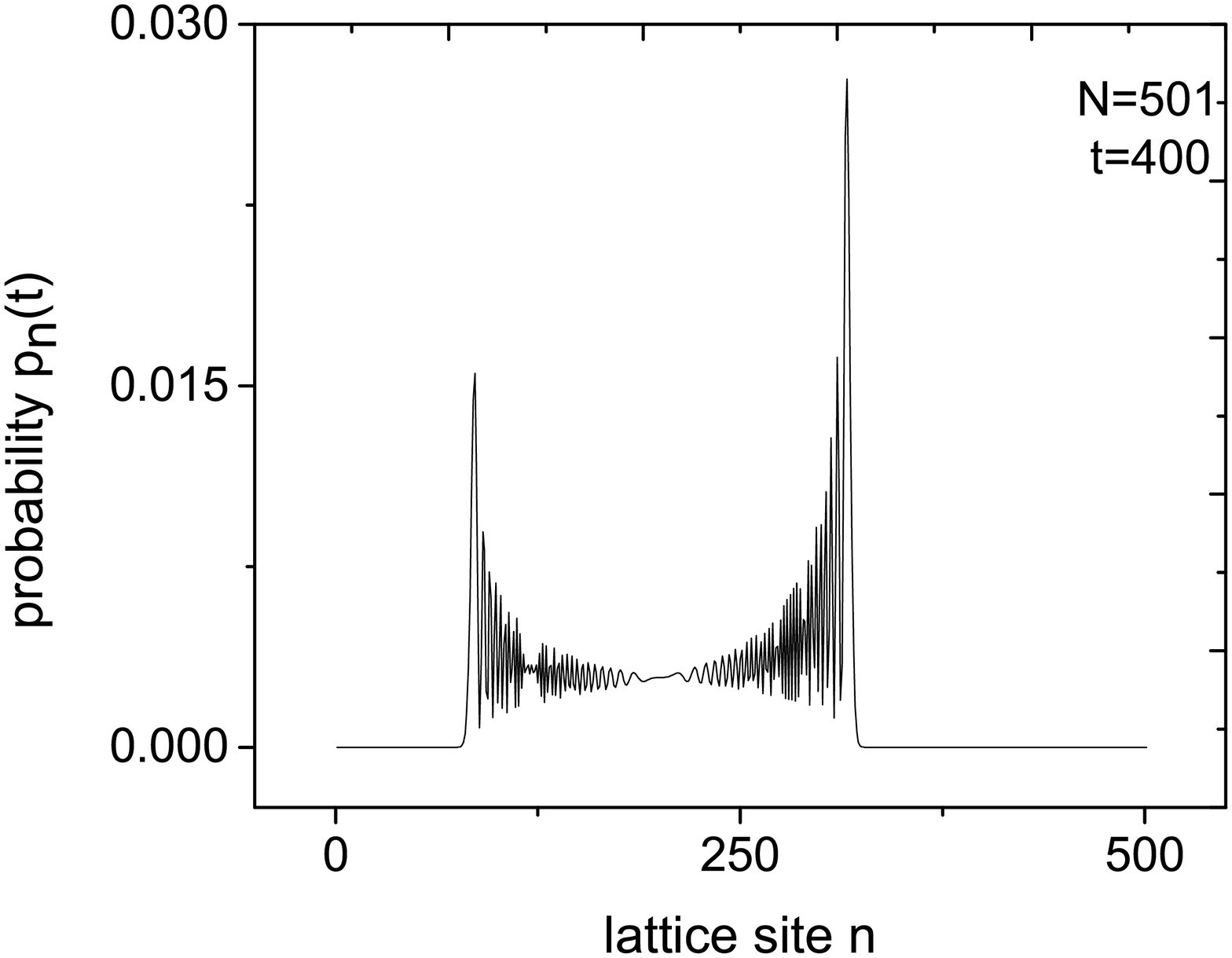}}
\caption{The probability distribution of position occupancy for the phase space walk at two different times shows a broadening very similar to the configuration space walk.}\label{fig:Probability distribution}
\end{figure*}
The phase space walk then shares the central property with the CSW in that the standard deviation increases linearly with time, as opposed to the classical diffusive $\sim \sqrt{t}$ \cite{kempe2003quantum, nayak2000quantum, childs2009universal}.
 The standard deviation, $\sigma(t)$, after time $t$ is found from, 
\beq
\sigma^2(t)=\sum_{n=0}^{N-1}n^2 p_n(t)-\left(\sum_{n=0}^{N-1}n\, p_n(t) \right)^2.
\eeq
The numerical evidence that the standard deviation grows linearly is given in Fig.~\ref{fig:Standard deviation}. The linear growth of the standard deviation is found in a certain range of the times, and before the onset of finite lattice size effects.  

\begin{figure}[]
\subfloat{\includegraphics[width = 3in]{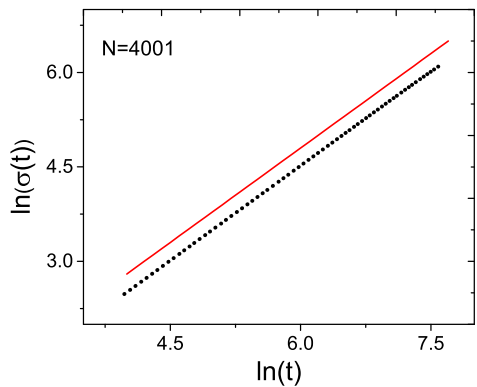}}
\caption{ Growth of the standard deviation of the walker position with time on a lattice of size $N=4001$ for the Hadamard phase space walk (dotted line). A solid line with unit slope is shown for comparison.}\label{fig:Standard deviation}
\end{figure}

However the standard deviation is only one aspect of the distribution of the walker and this can grow linearly even the case of trivial coins that produce no dispersion of the state.
Hence the participation ratio of the distribution, which is a measure of number of lattice sites that are significantly occupied at any particular time could be of interest. For example for $\theta=0$ the standard deviation grows as $t$, while the participation ratio is exactly $2$ at all times. This quantity does not seem to have been explored much in the context of the quantum walks (for some exceptions see \cite{lakshminarayan2003random,PhysRevA.94.023601,yalccinkaya2015two}), although it is widely used elsewhere to measure delocalization \cite{thouless1974electrons,
wegner1980inverse}.
The walker participation ratio at time $t$  is  
\beq
\label{eq:participation 4}
P(t)=\left(\sum_{n=0}^{N-1}p_n(t)^2\right)^{-1},
\end{equation}
and is such that $1\leq P(t) \leq N$, with the extremes indicating site localization and complete delocalization with equally likely site occupancies respectively. 

It is of interest first to find this for the classical walk. Assuming the case where the probability of traversing in both directions is same, the case that is relevant to compare with the Hadamard coin, the  inverse participation ratio after time $t$ is
\beq
P(t)^{-1}=\frac{1}{2^{2t}}\sum_{n=0}^t \binom{t}{n}^2=\frac{1}{2^{2t}}\binom{2t}{t}.
\eeq 
Using the Sterling approximation $n! \sim\sqrt{2\pi n}\left(n/e\right)^n$ it follows that the classical random walk participation ratio grows as
\beq
P(t) \sim \sqrt{\pi t}.
\eeq.
Hence the participation ratio shares with the standard deviation a slow normally diffusive growth. In contrast the quantum walks, both phase space version and the standard configuration space one produce distributions with lattice  participation ratio that grow faster:  $P(t) \sim t^\beta$, with $\beta \approx 0.825$, as Fig. \ref{fig:Participation Ratio} illustrates. Thus the quantum walker probability is considerably more delocalized but significantly does not seem to grow linearly in time. The power-law growth occurs in a time window, after an initial transient and before the finite size of the lattice affects the time-evolution. The latter time is shown as the divergence when $N$ is changed. As maybe expected, this time scales linearly,  $\sim \gamma N$, with $\gamma\approx0.25$ for PSW and $\gamma \approx 1.41$ for CSW.  This is of the order of the Heisenberg time and is naturally  much longer than the Ehrenfest time of the walk estimated below as $\sim \sqrt{N}$. It is interesting that the finite lattice effects happen much earlier for the phase space walk as compared to the CSW.
\begin{figure*}[]
  \subfloat{\includegraphics[width = 3in,height=2.3in]{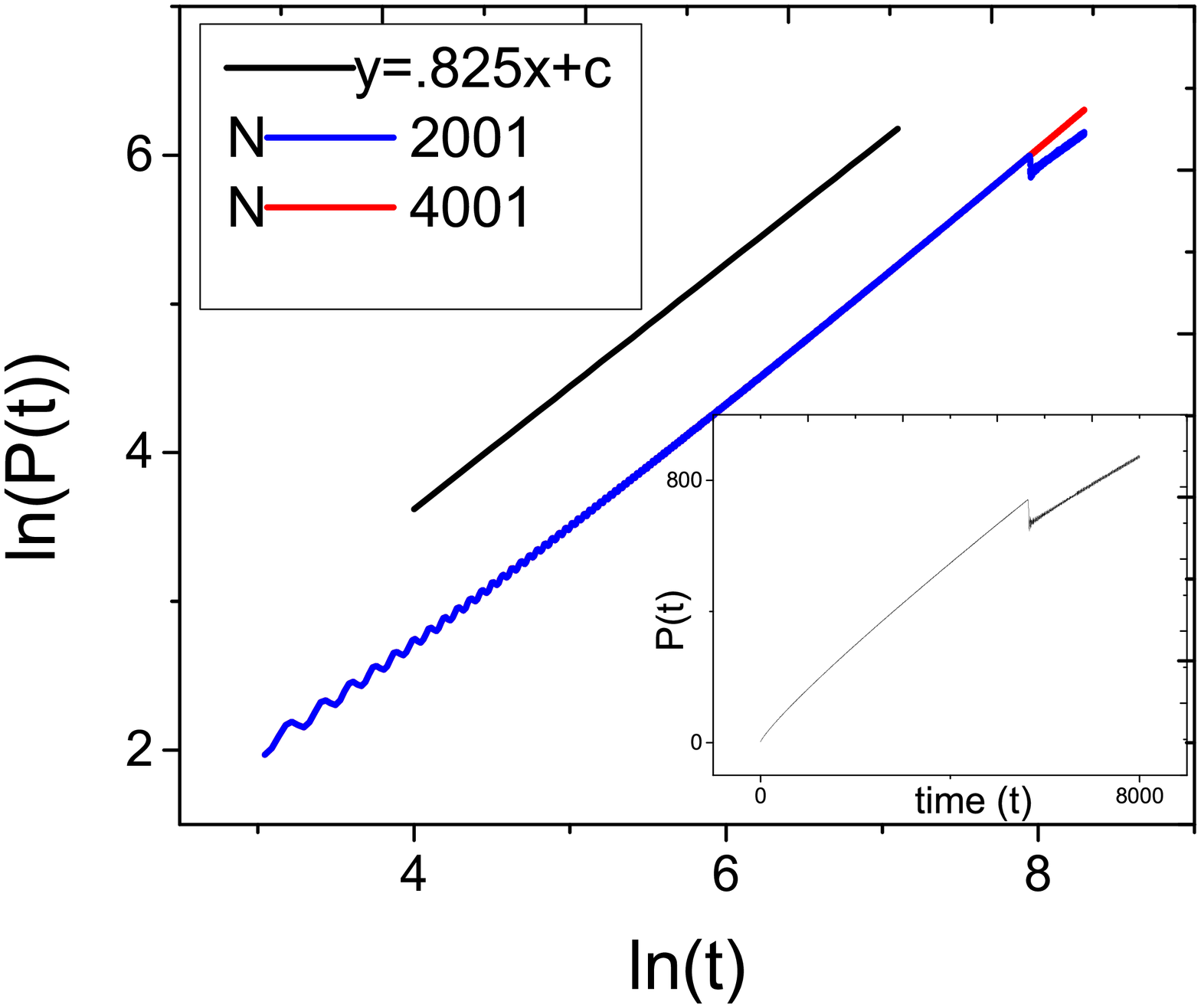}} 
    \subfloat{\includegraphics[width = 3in]{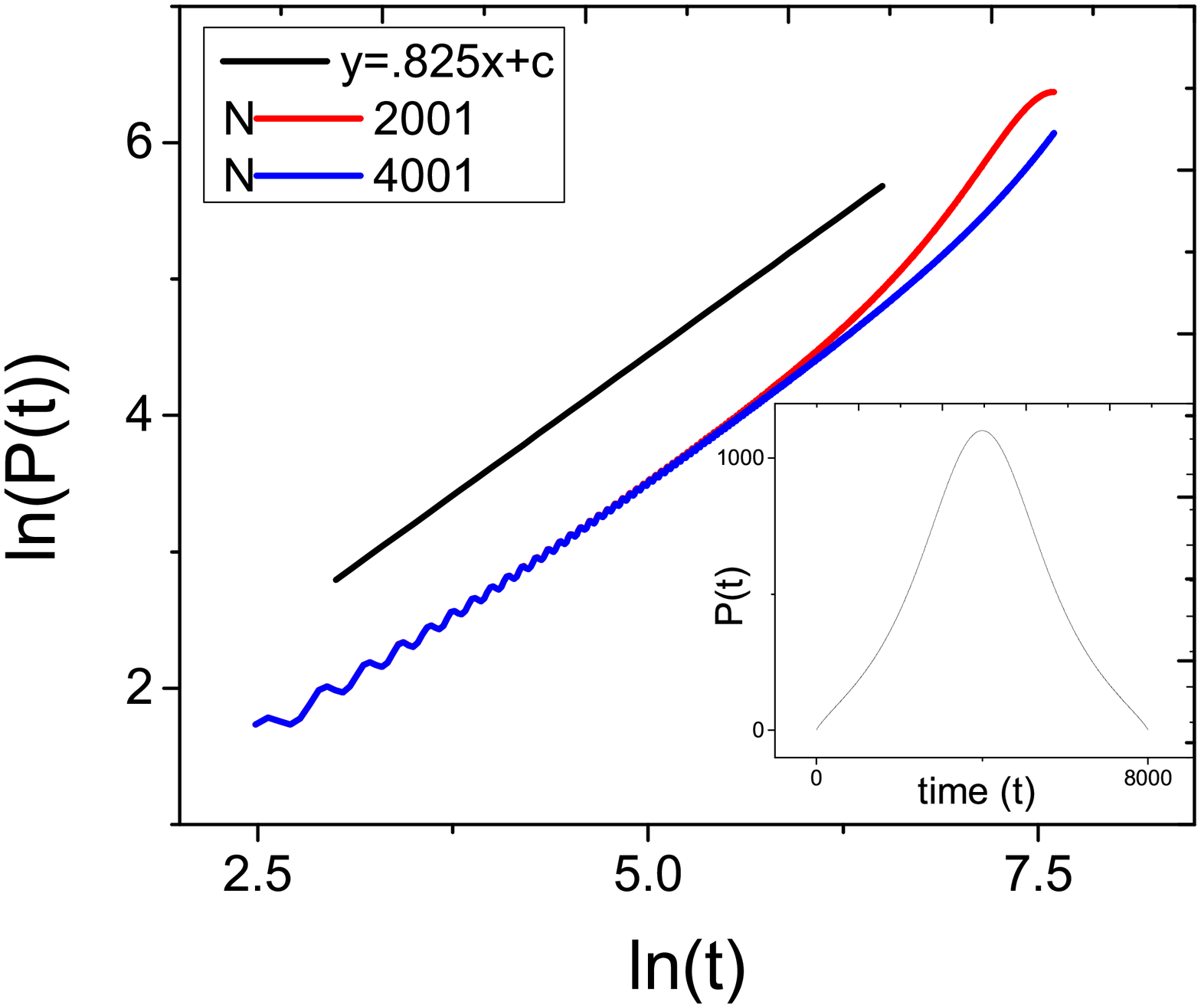}}
    \caption{Illustration of the power-law growth of the walker participation ratio for the configuration space walk (CSW, left) and phase space walk (PSW, right). Lattices of sizes $N=4001$ and $2001$ are used. The straight lines on the top have a slope of $0.825$. The insets are over longer times and show the different fates of the CSW which fluctuates around a saturation value (not seen in the figure) and the exact periodicity of the PSW with a period $2N$.}\label{fig:Participation Ratio}
  \end{figure*}
  
  The dependence of the participation ratio on the angle $\theta$ parameterizing the coin operator at various times for the CSW and PSW  is shown in Fig.~(\ref{fig:participation theta2}). This shows a marked dependence on the coin dynamics and also that the angle at which the maximum value of the participation ratio 
occurs is dependent on the time. In particular it is not true that the Hadamard coin is singled out, except at $t=N/2$ when the participation ratio is a maximum for this case.
However at other times there is a shift away from the Hadamard that corresponds to maximum delocalization of the walker. This is true for both the CSW and the PSW as shown in the figure. 

After time steps of the order of $N$, for PSW there is very large delocalization when the value of the participation ratio gets close to $N$ itself. This is seen when the coin angle $\theta \approx 0$. In contrast, this is not seen in the CSW and the maximum value of the participation ratio after $N$ time steps occurs for  $\theta \approx \pi/3$. Both kinds of walks have qualitatively similar behaviors for times well before the finite size effects start, the Heisenberg time, or  when the power-law growth exists. Thereafter quantum interference effects due to finite boundaries have very different effects on the walkers.

  \begin{figure*}[]
  
    \subfloat{\includegraphics[width = 3.2in]{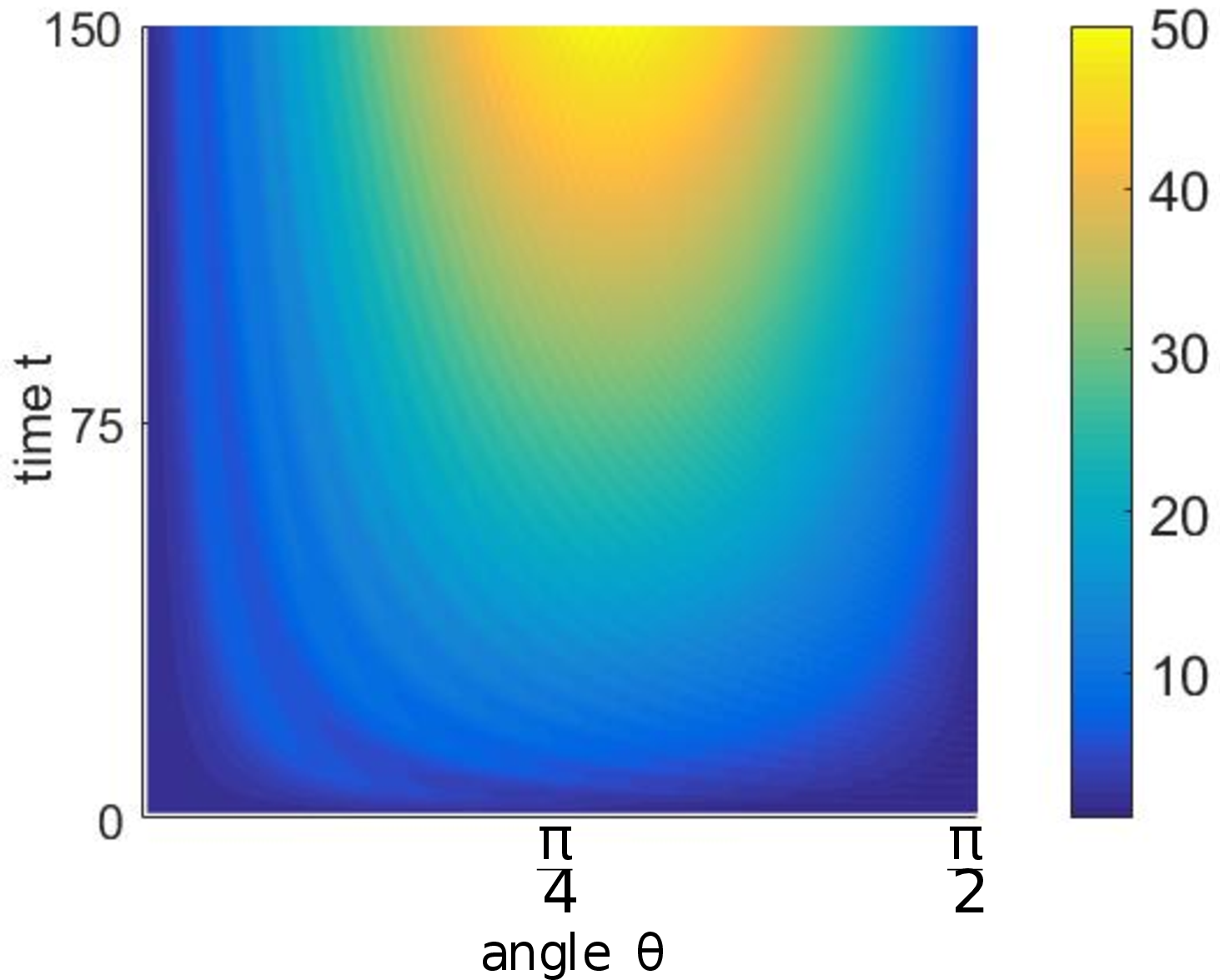}}
    \subfloat{\includegraphics[width = 3.2in]{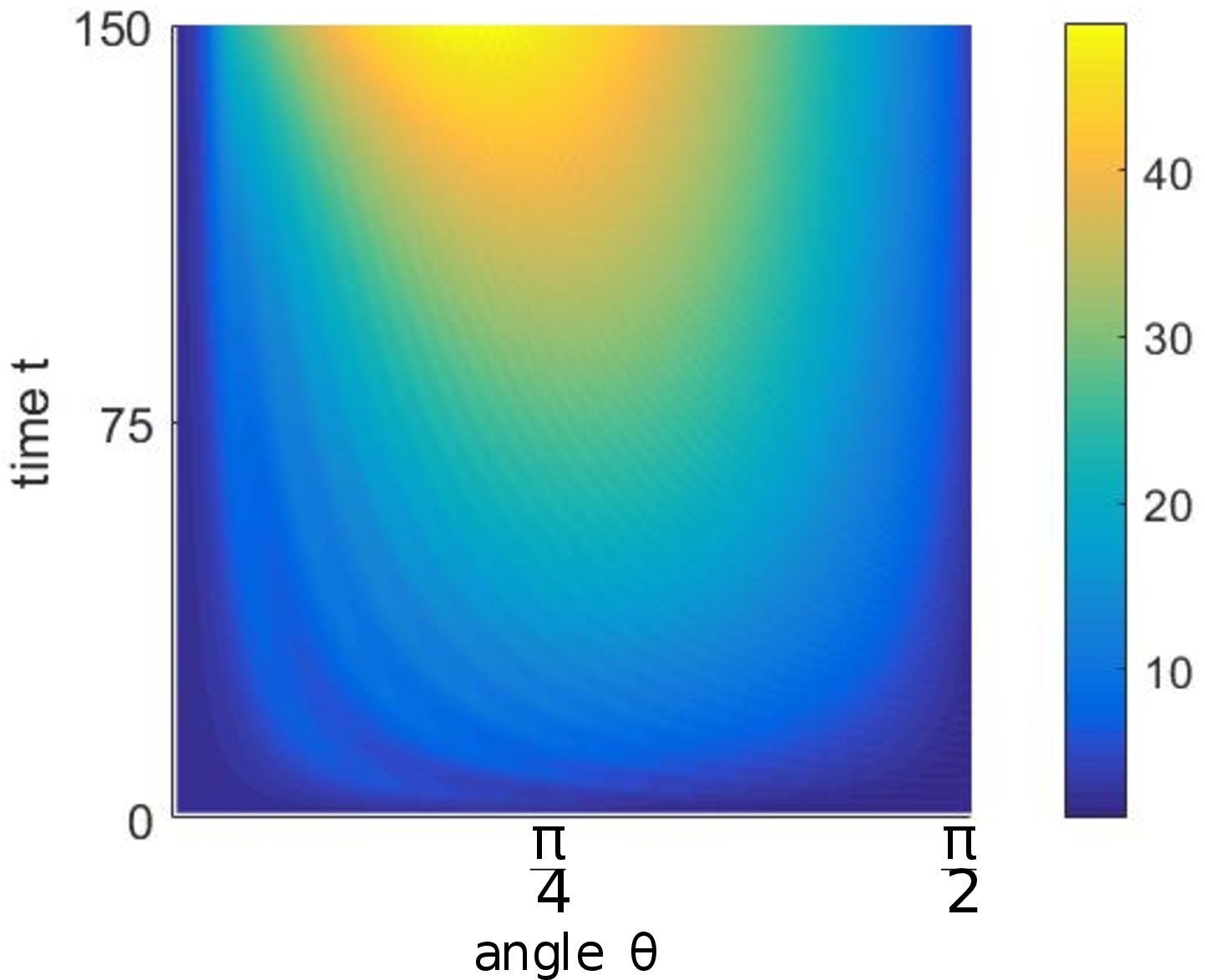}} \\
     \subfloat{\includegraphics[width = 3.2in]{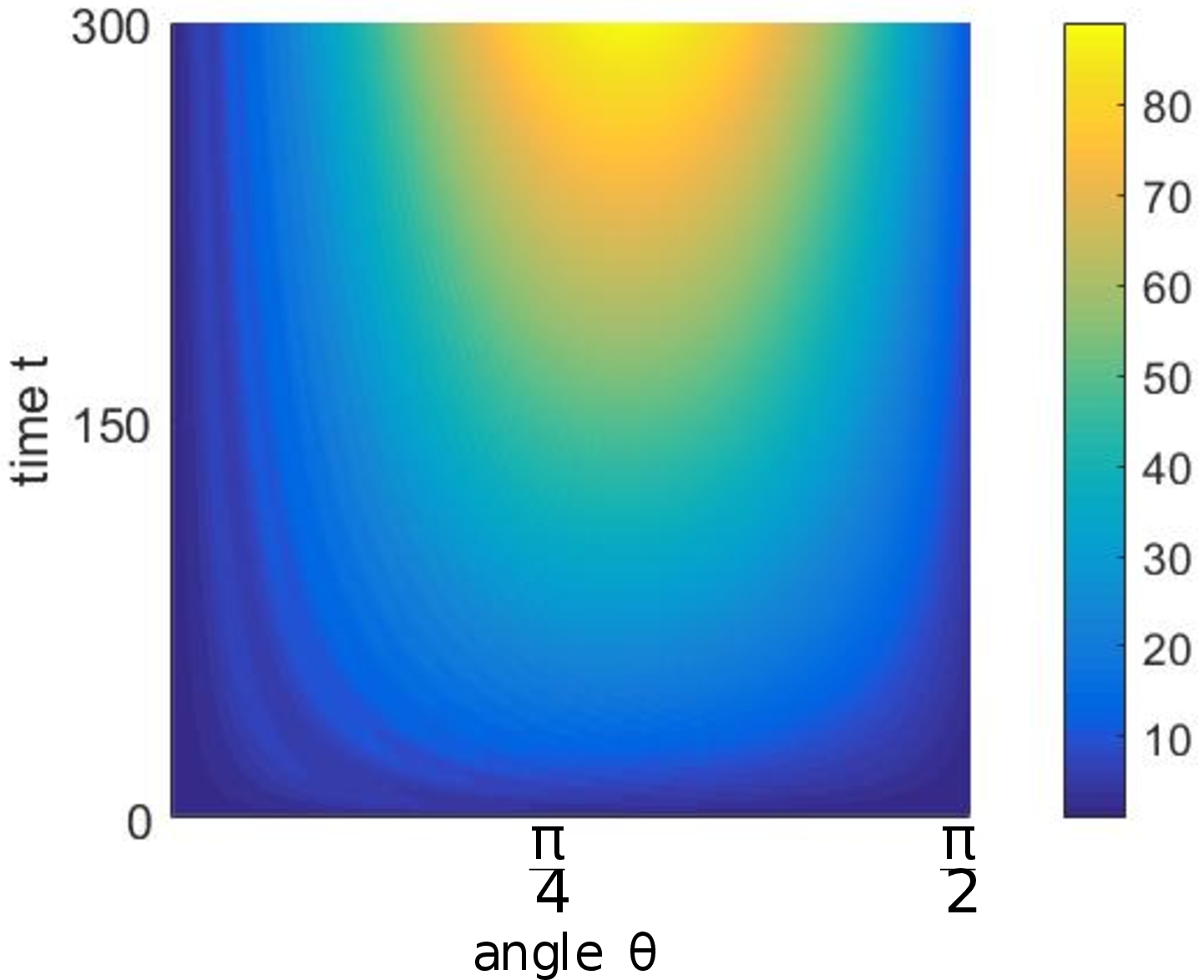}}
    \subfloat{\includegraphics[width = 3.2in]{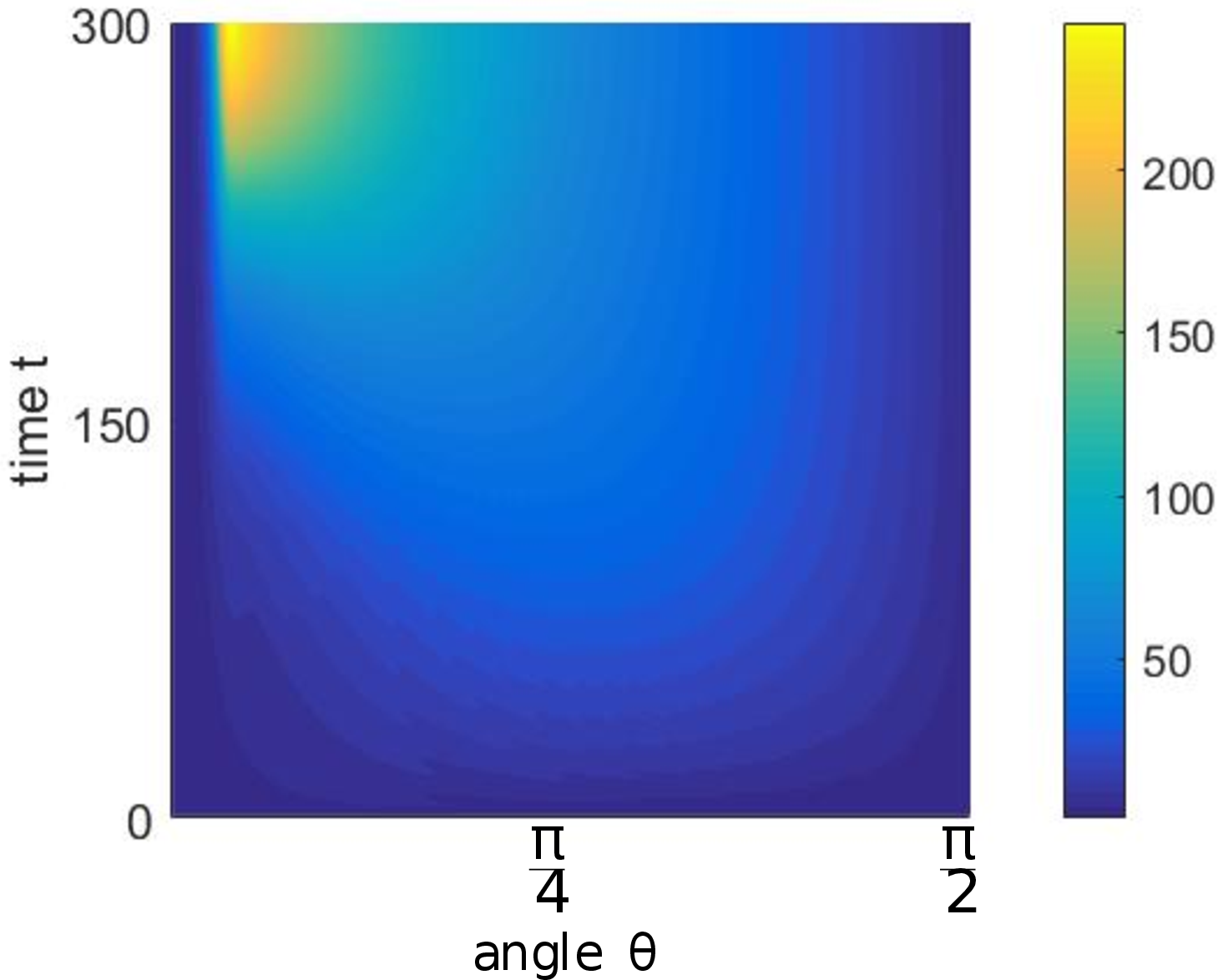}} 
   
    \caption{The dependence of the participation ratio on time $t$ and the coin angle $\theta$ for the CSW (left) and PSW (right). Here $N=301$  and the  initial state is $|{0}\kt\otimes |{0}\kt$. The top figures evolve upto $=150 \approx N/2$ steps and the bottom ones are upto $301 =N$. }\label{fig:participation theta2}
    \end{figure*}

   While taking the position or momentum state as the initial state of the walker, either momentum or position is respectively completely delocalized. As the walk is in a phase space it is natural therefore to consider the fate of walkers that are localized in phase space, therefore we turn to the case of an intial coherent state for the walker. A recent work \cite{zhang2016creating} studies variously delocalized Gaussian states in one-dimensional configuration space walks and find ``cat states" forming. We find that the PSW also produces such cat-states in phase space and that the walk has very regular structures in the coherent state case.

\subsection{Evolution of coherent states}
From the  reduced density matrix of the walker $\rho_w(t)$ the Husimi distribution \citep{PhysRevLett.55.645}, a psuedo-probability phase space distribution, can be constructed as,
\begin{equation}
\label{eq:Husimi 2}
W_{\rho_w}(p,q,t)=\bra{(q,p)}\rho_w(t)|{(q,p)}\kt.
\end{equation}
Where $|{(q,p)}\kt$ is a coherent state localized at $(q,p)$ in the phase space and can be for example constructed for the toral phase space using the Harper Hamiltonian's ground state as the fiducial state $|(0,0)\kt$ \citep{saraceno1990classical}. Here $0\leq p,q <N$ are the discrete pseudo-phase space variables, covering the torus. In the following the walker starts at the origin and with zero momentum, that is the fiducial state $|(0,0)\kt$ itself.
Shown in Fig.~\ref{fig:Evolution of PSW with coherent State} is the Husimi representation of the walker state in $U_{psw}^t |(0,0)\kt \otimes |\phi\kt$ for two initial states of the coin $|\phi\kt$, symmetric and asymmetric \citep{kempe2003quantum}: 
\begin{equation}
\begin{aligned}
\label{eq: coin-state}
|{\phi_{sym}}\kt=\frac{1}{\sqrt{2}}\left(|{0}\kt+i\|{1}\kt\right),\,\,
|{\phi_{asym}}\kt=|{0}\kt.
\end{aligned}
\end{equation}

The walker starting at the origin in phase space evolves approximately classically initially in the sense that it spreads out along the $p+q=t$ line but with a width
that comes from uncertainly in phase space. This phase lasts for about a time $\sim \sqrt{N}$, after which a ``split" into two peaks gets well-defined and these two 
separate and move into phase space. This is the creation of ``cat-states" and is completely non-classical. This phase is shown in Fig.~\ref{fig:Evolution of PSW with coherent State}, where the two peaks at various times are seen. It is found that  for the symmetric coin-state the phase space representation  has two symmetric  peaks about the line $p=q$. However for  the asymmetric case even though there are two peaks, they are of unequal magnitude. After a time $N$ the two cat-states merge once again at the center of the torus and then continue to grow apart till time $2N$ when due to the exact periodicity of the walk, the initial state is recovered.

\begin{figure*}
  \subfloat[Symmetric Initial coin-state]{\includegraphics[width = 3.1in]{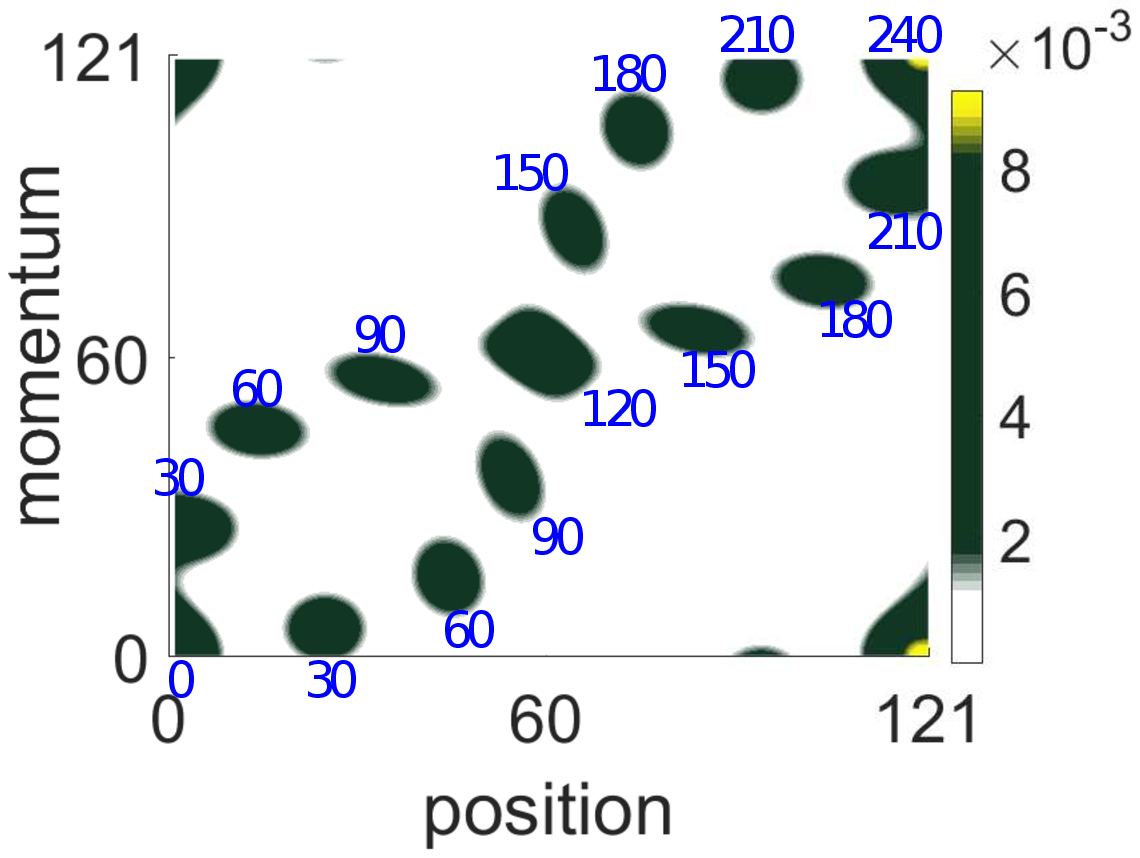}} 
\subfloat[Asymmetric Initial coin-state]{\includegraphics[width = 3in]{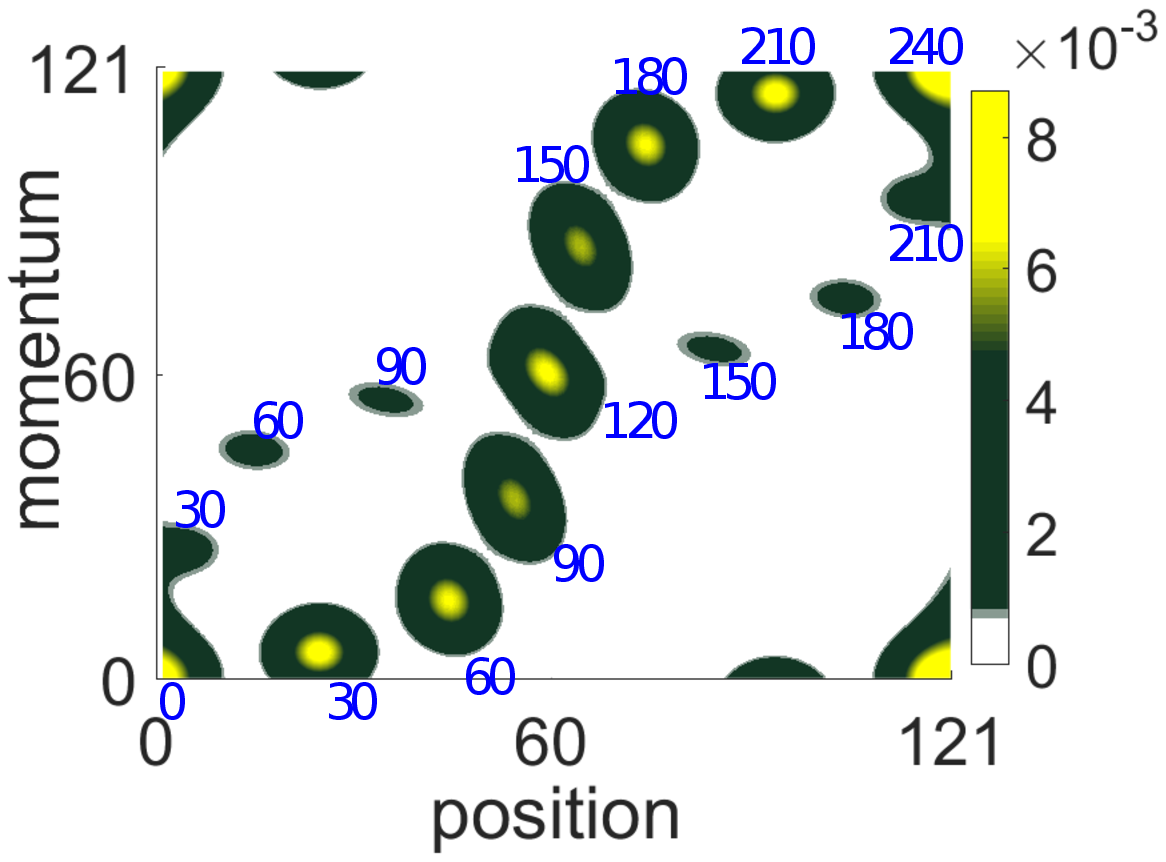}}
  \caption{The Husimi distribution of the time evolution of the walker state $|{(0,0)}\kt$  for the symmetric coin-state (right) and the  asymmetric coin-state (left) on a lattice of size  $N=121$. The Husimi distribution is shown on the same figure at various times indicated therein.}
\label{fig:Evolution of PSW with coherent State}\end{figure*}

Cat-states find applications in circuit QED\citep{zhang2017universal,girvin2017schrodinger}, quantum information processing\citep{gilchrist2004schrodinger}and quantum computation\citep{MIRRAHIMI2016778}  Hence the creation of cat-states are of fundamental importance in physics. Unlike the case of the CSW for PSW there is no translational invariance in the problem  and the cats are formed in phase space with varying momenta, whereas in the case of CSW all cat-states formed have the same momentum \cite{zhang2016creating}.  

The Husimi representation of the time evolution of the walker starting at the phase space origin can be understood as it resembles that of the walker reduced state of the eigenvectors of $U_{psw}$ with eigenvalues $\pm 1$, see Figs.~ (\ref{fig:Evolution of PSW with coherent State},\ref{fig:eigenvector1}). Only the Husimi of the (walker reduced) eigenvector with eigenvalue $+1$ is shown in the latter figure, that of the state with eigenvalue $-1$ is obtained by a reflection about the $q=p$ line. Indeed it is found that the  eigenvectors with eigenvalues close to $\pm 1$ contribute most to the initial state $|(0,0)\kt |0\kt$, and hence this structure dominates the time evolution. The probability distribution of the eigenvector with eigenvalue $+1$, after tracing out the coin, in the momentum and position basis, is also shown in  Fig.~(\ref{fig:eigenvector1}). This indicates that they are simply mutually shifted from each other. For an eigenvalue $\omega^{k/2}$  the position basis distribution has a maximum at $n=(k-N)/2\; \text{mod} \, N$ while the momentum basis distribution peaks at $n=(N-k/2)\; \text{mod} \,N$. This also follows from the momentum basis representation of the eigenstates that shows an interesting duality. The eigenvectors in the walker's momentum basis, is given up to normalization by $|\phi_k\kt=\sum_n (\tilde{a}_n(k) |\tilde{n}\kt |0\kt+ \tilde{b}_n(k) |\tilde{n} \kt |1\kt)$, where
\begin{subequations}
\label{eq:eigenvector21}
\begin{align}
\tilde{a}_n(k)=&b^{-1}_n(2N-k)=-\omega^{\frac{n(n-1)}{2}}\tan \theta\, \dfrac{\left(\sec\theta\,\, \lambda_k^{-1};\omega^{-1}\right)_n}{\left(\sec\theta\, \,\lambda_k;\omega\right)_{n+1}},\label{subeqn:an1}\\
\tilde{b}_n(k)=&a^{-1}_n(2N-k)=\omega^{-\frac{n(n-1)}{2}}\frac{\left(\sec\theta\,\,\lambda_k^{-1};\omega^{-1}\right)_n}{\left(\sec\theta\,\, \lambda_k;\omega\right)_{n}}\label{subeqn:bn1}.
\end{align}
\end{subequations}
Thus interestingly the $a_n$ that appear in the momentum representation are related to the $b_n$ in position and are not pure phases, while the $b_n$ of the momentum basis are pure phases.
   
\begin{figure*}
\subfloat{\includegraphics[width = 3in]{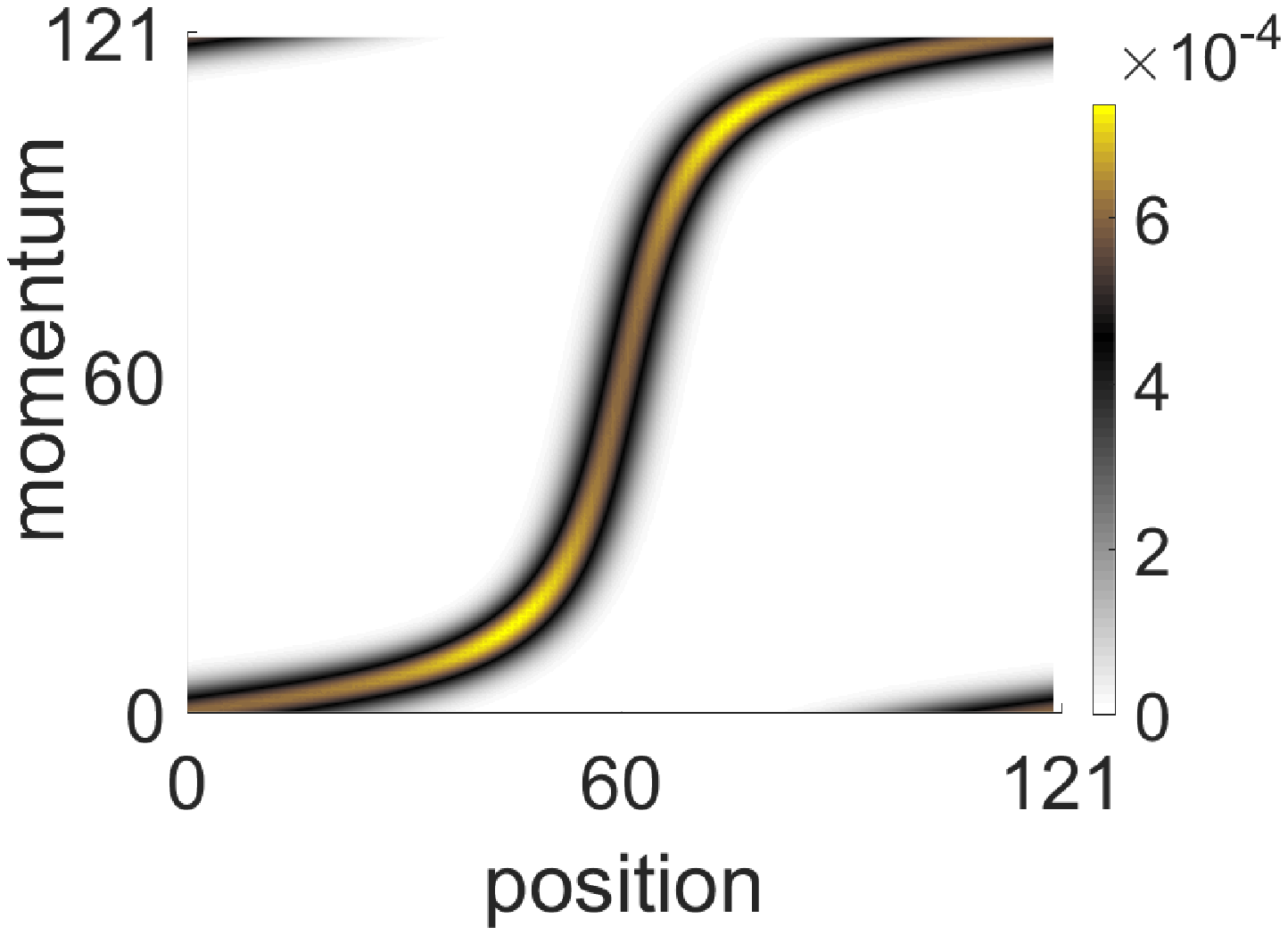}}
\subfloat{\includegraphics[width = 2.8in]{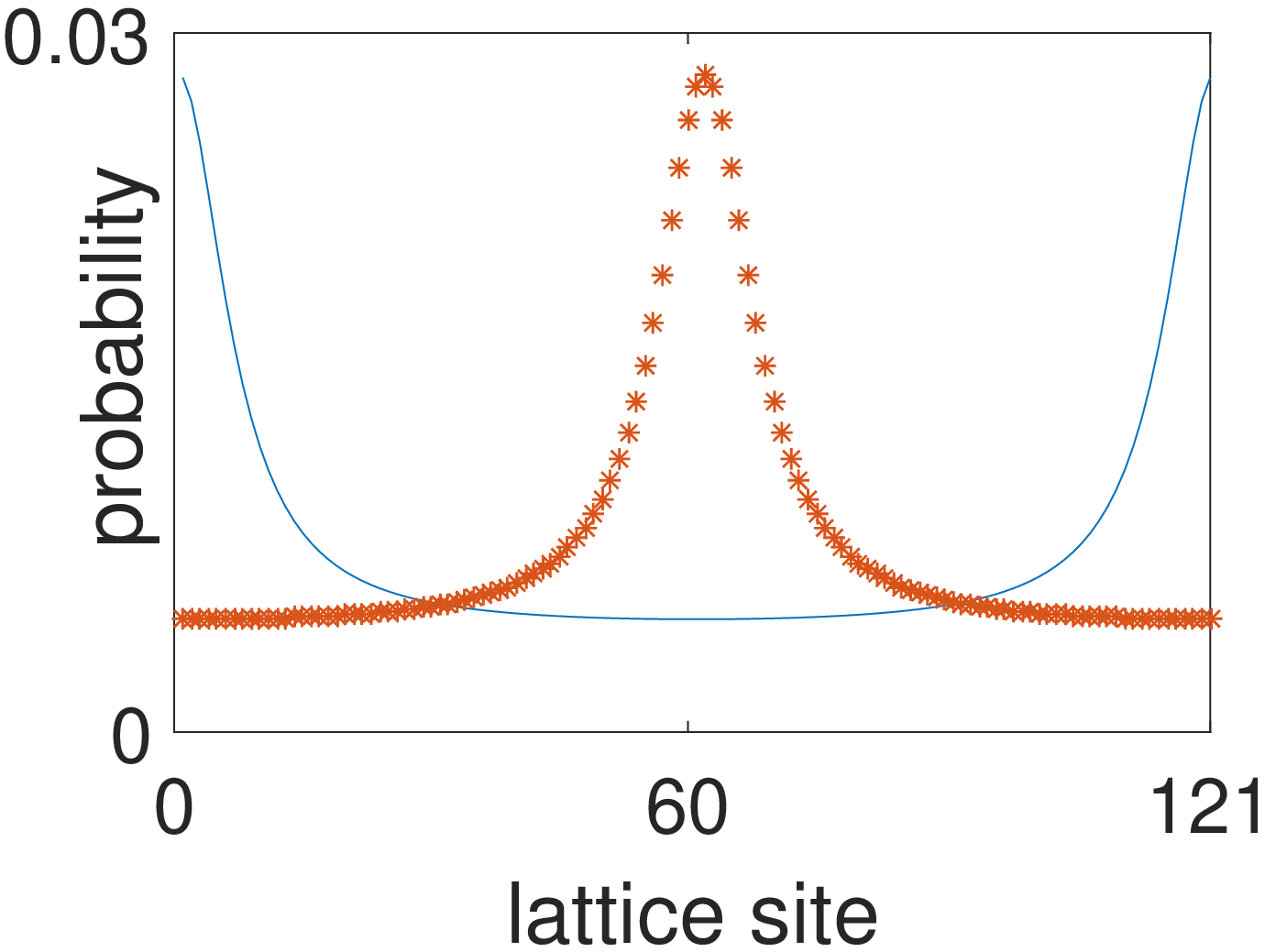}}
\caption{Husimi distribution of the walker reduced eigenvector with eigenvalues $+1$(left) for the case $N=121$. The probability distribution of the same is shown in the position basis (points) and the momentum basis (solid line).}\label{fig:eigenvector1}
\end{figure*}
In \cite{zhang2016creating} the growth of entanglement with time has been studied as an indicator of the formation of cat-states, with the entanglement nearly saturating with their formation.  The entanglement growth in the PSW is shown in Fig.~(\ref{fig:Entanglement}), where one finds a similar behavior, but also that the saturation value depends on the initial state. For the case of the symmetric coin-state the maximum value of entanglement (von Neumann entropy) appears to be very close to $1$, and hence the coin and walker get maximally entangled. However for the asymmetric case this value is not achieved, and moreover the entanglement growth is not monotonic, nevertheless a saturation seems to happen at large $N$ for to an entropy of about $0.6$.

\begin{figure*}
   \subfloat{\includegraphics[width =3in]{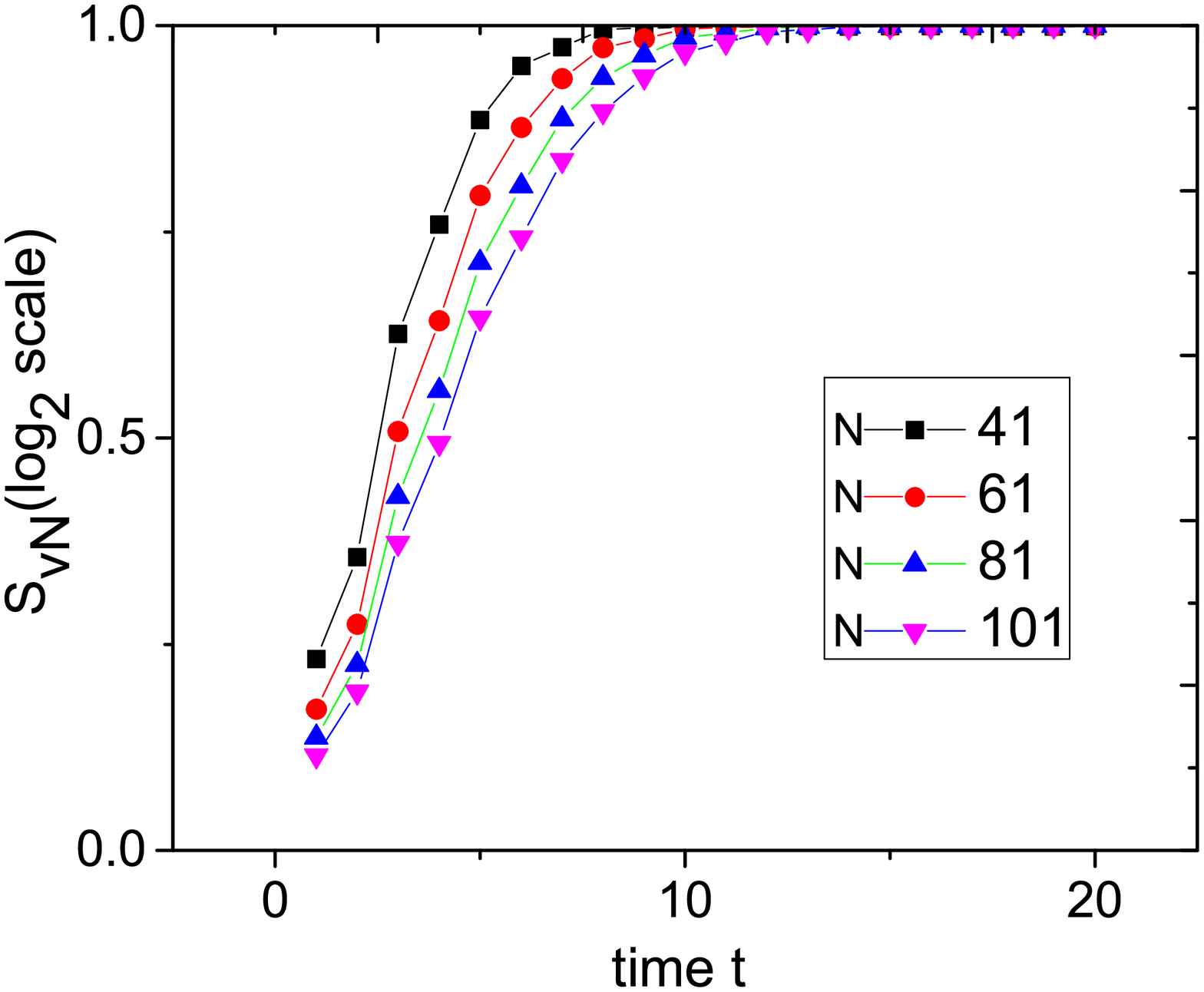}}
\subfloat{\includegraphics[width =3in]{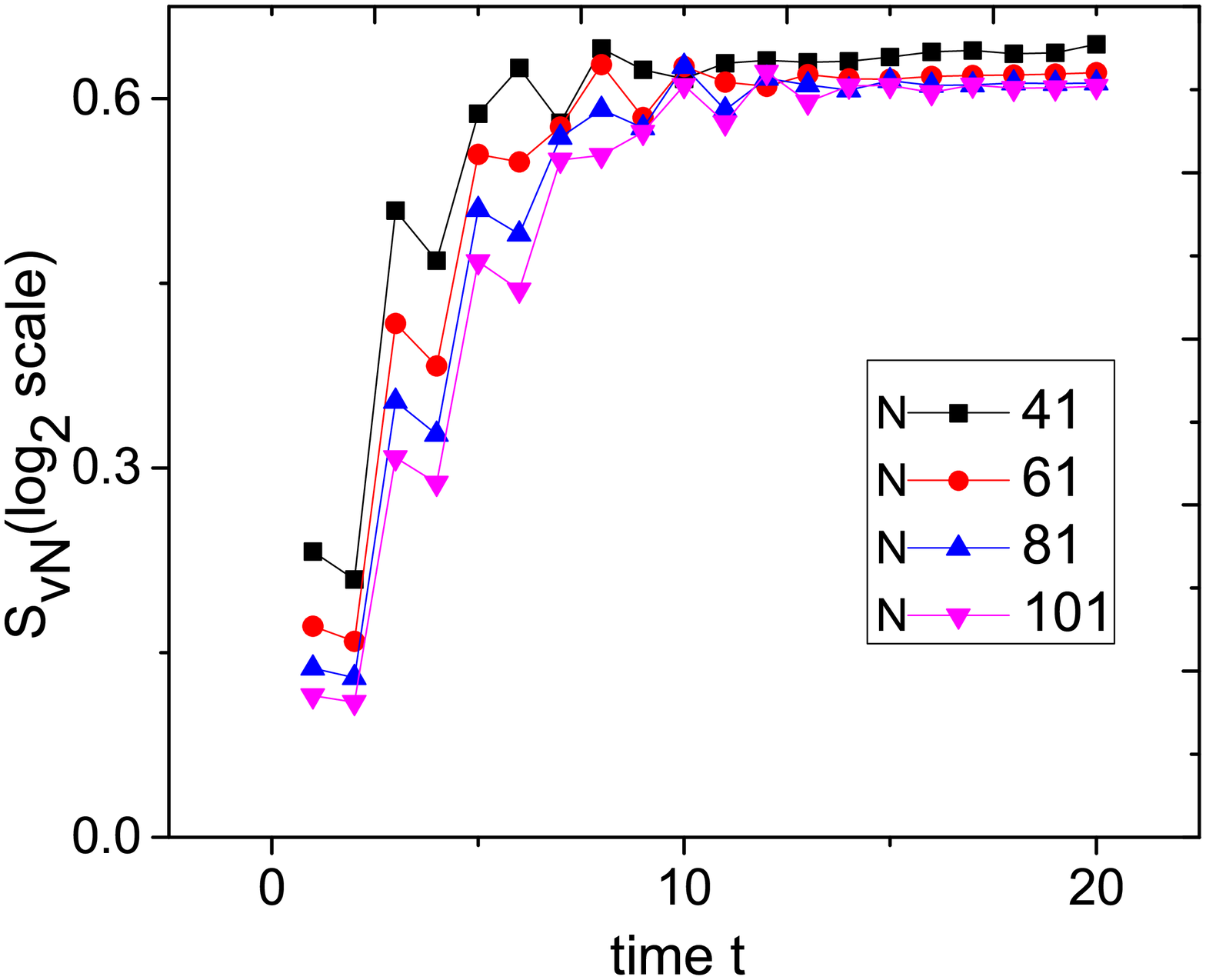}}
\caption{Growth of the coin-walker entanglement with time for a symmetric initial coin state (left) and asymmetric coin state (right) and an initial coherent state for the walker localized at the phase space origin. For the symmetric as well as asymmetric cases the entanglement reaches a saturation value after some time, however  for the symmetic case entanglement achieves nearly the maximum possible  value($1$) but for  asymmetric case the saturation value depends on the dimension on the lattice under consideration.}\label{fig:Entanglement}
\end{figure*}

The time to attain the maximum value of entanglement is an indicator of the onset of quantum interference effects and is therefore an Ehrenfest time $t_E$ of the quantum walk. In the case of symmetric initial states, the time taken to reach within $10^{-3}$ of the maximum entanglement is taken as $t_E$. In Fig.~(\ref{fig:Classical}) this is shown for different lattice dimensions $N$ for the PSW. This suggests the growth $t_E \sim \sqrt{N}$, a feature that we also verified holds for the CSW and for asymmetric initial coin states. Thus there is an algebraic growth of the Ehrenfest time and this is large in comparison to quantum chaotic maps on the torus that support a logarithmic Ehrenfest time $\sim \ln(N)$. Thus this also illustrates the lack of quantum chaos in quantum walks, at least of the kind considered here. This is also consistent with the behavior of out-of-time-ordered correlators of the quantum walk that increases as $t^2$ rather than exponentially \cite{sivaarul}.

\begin{figure}
\subfloat{\includegraphics[width = 3 in]{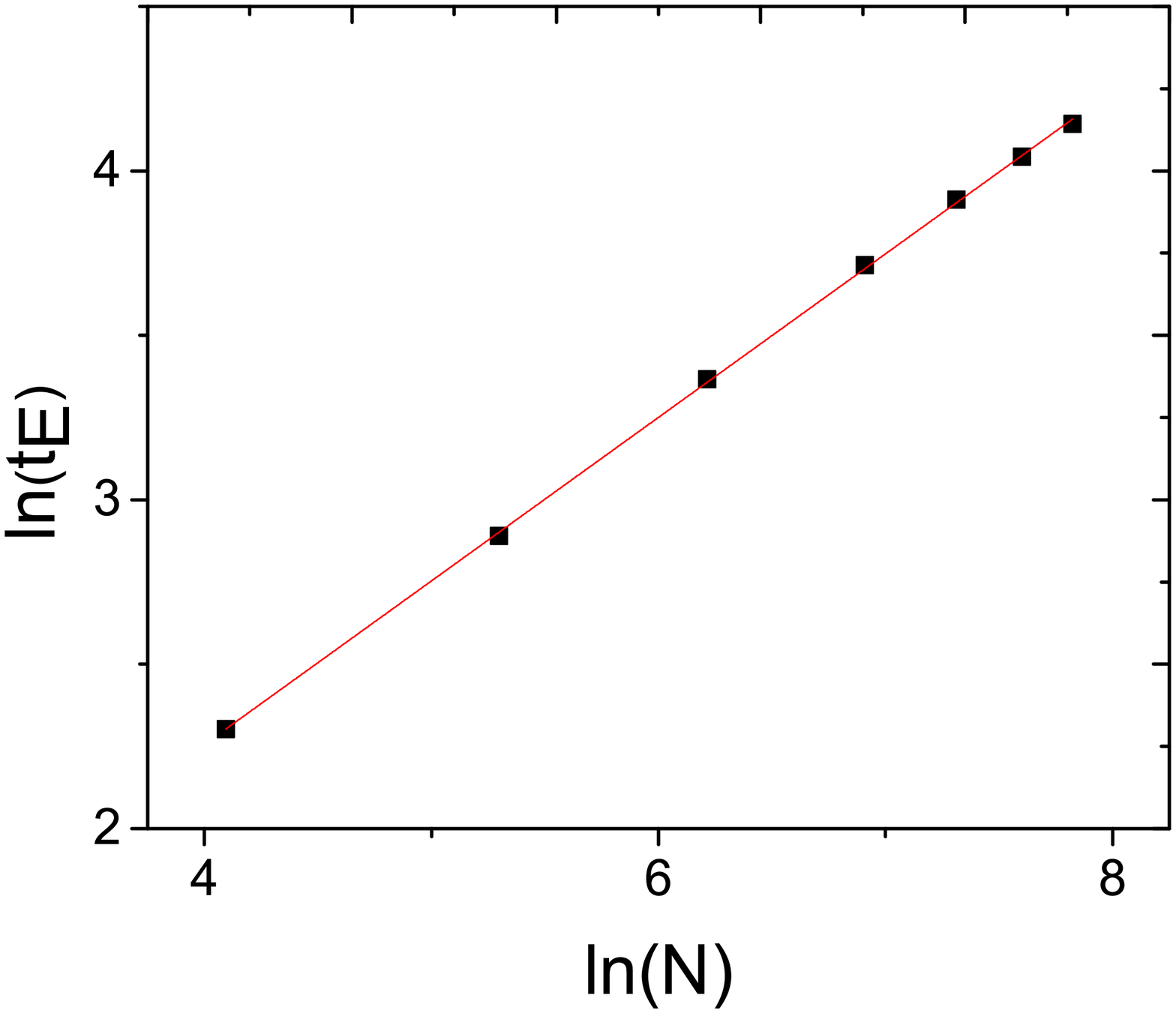}}
\caption{The Ehrenfest time, indicating the onset of quantum effects, {\it vs} lattice dimensions, showing that it grows as $\sqrt{N}$}\label{fig:Classical}
\end{figure}
\section{Summary and conclusion}
This paper introduced a quantum walk in toral phase space (PSW), and the walker could either change her position or her momentum depending on a coin toss.
The unitary operator corresponding to this thus has simply either translations in positions or boosts in momentum. A detailed study was possible as it turned out that the spectra is exactly solvable. The eigenangles are equally spaced on the circle for the case of odd dimensionality (and there is a very similar structure for the even) and the eigenvectors can be written either in terms of the $q-$Pochhammer symbols or in terms of elementary functions. This is true for an entire family of coin operators and interestingly the eigenvalues do not change (for odd lattice dimensionality) within this coin family. The eigenvectors also interpolate within this family in an interesting manner and it is possible to also find exactly the entanglement in them. 

The time-evolution of the walker with an initial state that is either site localized or a coherent state that is phase-space localized was both considered. The PSW and CSW share many common features, for example the growth of the standard deviation and participation ratio occur with the same power laws. The participation ratio which is a measure of delocalization of the walker has, to our knowledge, not been explored in the context of quantum walks. This quantity also grows at a power law that is higher than the classical walker and hence the quantum walker does get more delocalized with time. The delocalization of the walker at any given time as measured by the participation ratio has an interesting dependence on the angle parametrizing the biasedness of the coin. In particular, the most unbiased coin as represented by the Hadamard is not necessariy the most delocalizing one.

When the initial state of the walker is a coherent state in phase-space ``cat-states" are formed in the case of the PSW, and may lead to interesting consequences. Such cat-states have also been reported recently for the CSW in \cite{zhang2016creating}. Their formation is concomitant with the generation of maximum entanglement between the walker and the coin. This happens at a time $t_E \sim \sqrt{N}$ and maybe considered as an Ehrenfest time of the walk. We have explored, but not reported, the case of higher dimensional coins and more possibilities for the walker in phase-space, a natural case is when both the position and momentum can also decrease by a unit and the coin is four dimensional. In this scenario cat-states localized at more than two distinct parts of the phase space were observed.  

The similarity of the walk studied here with other walks including the electric-walk has been noted. One future direction may be in studying such phase-space walk without a toral phase space and the effect of an irrational multiple of $2 \pi$ for the angle $\alpha$ in 
the definition of the phase-space walk in Eq.~(\ref{eq:psw}). Others naturally include exploration of a more complete family of coins (we have only considered a subset of $SU(2)$), higher dimensional coins and phase-spaces and possible experimental realizations. We note that one realization with ion traps is already studied as a 
phase-space walk in the sense that it consisted of walk among a one-dimensinal lattice of coherent states \cite{schmitz2009quantum}. As the elements of the walk considered here involve only the fundamental acts of translations in position and momentum, it is conceivable that it is realized in many different setups.

\section{Acknowledgment}
We thank C. M. Chandrashekar (I.M.Sc., Chennai), Sandeep K. Goyal (IISER Mohali) and Prabha Mandayam (I.I.T. Madras) for valuable discussions and remarks. SO would like to specially thank Prabha Mandayam for funding from Department of Science and Technology, via INSPIRE Project No. PHY1415305DSTXPRAN, hosted at I.I.T. Madras.

\appendix
 \section{Evaluation of a sum in the normalization of eigenstates}
The constant appearing in the normalization of an eigenstates with eigenvalue $\lambda_k$ and lattice dimension $N$ is, $C_N(k)= \sum_{n=0}^{N-1} \left(a_n(k)a_n^*(k)+b_n(k)b_n^*(k)\right).$ Now using the Eq.~\eqref{eq:b-a:relation} and  $\vert a_n(k)\vert=1$  yields, 
 \begin{equation}
  \label{eq:normalization2}
C_N(k)=N+ \sum_{n=0}^{N-1}\dfrac{\tan^2 \theta\,}{1+\sec^2\theta\,+\sec\theta\,(\lambda_k\omega^{-n}+\lambda_k^{-1}\omega^n)}.
 \end{equation}
Let the second term in the RHS of the above equation  be $I$. 

For odd $N$ using  Eq.~\eqref{eq:eigenvalue odd} and the Poisson summation formula (for example see \citep{benedetto1997sampling} for a extensive discussion) gives
\beq
\begin{aligned}
\label{eq:I1}
I=& \sum_{n=0}^{N-1} \dfrac{\tan^2 \theta}{1+\sec^2\theta+2 \sec \theta \cos(2 \pi(n-\frac{k}{2})/N)}\\
=&N\int_{0-\epsilon}^{1-\epsilon }\dfrac{\tan^2 \theta \, dx}{1+\sec^2\theta+2 \sec \theta \cos(2 \pi(x-\frac{k}{2}))}\\
+2 N&\sum_{m=1}^{\infty}\int_{0-\epsilon}^{1-\epsilon}\dfrac{\tan ^2 \theta \cos(2\pi m N x)\, dx}{1+\sec^2\theta+2\sec\theta\cos(2 \pi (x-\frac{k}{2}))}.
\end{aligned}
\eeq
However, irrespective of the parity of $k$, 
\beq
\int_{0}^{1}\left(\dfrac{\tan^2 \theta \, dx}{1+\sec^2\theta+2 \sec \theta \cos(2 \pi (x-\frac{k}{2}))}\right)=1,
\eeq and
 \begin{equation}
\label{eq:eq2}
\begin{split}
\int_{0}^{1}\dfrac{\tan ^2 \theta \cos(2\pi m N x) \, dx}{1+\sec^2\theta+2\sec\theta\cos(2 \pi (x-\frac{k}{2}))}\\=(-1)^{(k+1)N}\cos^{Nm} \theta.
\end{split}    
    \end{equation}
Hence 
\[ I=N \left(-1+\frac{2}{1+(-1)^k \cos^N\,\theta} \right),\]
 which in turn yields,
\begin{equation}
\label{eq:nn1}
C_N(k)=\frac{2N}{1+(-1)^k \cos^N\,\theta}.
\end{equation} 

For even $N$  using Eq.~\eqref{eq:gh}, the equivalent of Eq.~\eqref{eq:I1} is, 
\begin{equation}
I'= N +2 \sum_{m=1}^{\infty}\int_{0}^{1}\dfrac{\tan ^2 \theta \cos (2\pi m N x)\cos(mN\alpha)  \, dx}{1+\sec^2\theta+2\sec\theta\cos 2 \pi x}.
\end{equation} 
However  using Chebyshev Polynomial of first kind $(T_m)$ and Eq.~\eqref{eq:gh}, $\cos(mN\alpha)= \cos( m \cos^{-1}(\cos^N\theta))=T_m(\cos^N\theta)$ and hence, 
\begin{equation}
I'= N +2 N \sum_{m=1}^{\infty}\int_{0}^{1}\dfrac{\tan ^2 \theta \cos (2\pi m N x) T_m(\cos^N\theta) \, dx}{1+\sec^2\theta+2\sec\theta\cos 2 \pi x}.
\end{equation} 
Now using \eqref{eq:eq2} and $T_{0}=1$ yields, $
C_N(k)=2N \sum_{m=0}^{\infty}T_n(\cos^n \theta)\cos^{mN}\theta$ and using the property of ordinary generating function of Chebyshev Polynomial of first kind \citep{mason2002chebyshev} given as, $
 \sum_{t=0}^{\infty}T_n(x)t^n=(1-tx)/(1-2xt+x^2) $ yields,  
  \begin{equation}
C_N(k)=2N\frac{1-\cos^{2N} \theta}{1-2\cos^{2N} \theta+\cos^{2N} \theta}=2N.
\end{equation}
Hence the normalization constant is surprisingly simple and is independent of eigenvalues. Note that for large enough $N$, even if $N$ is odd $C_N(k) \approx 2 N$.

\label{app:norm}
\section{ Evaluation of a sum appearing in the eigenvector reduced density matrices}
The sum to be evaluated in Eq.~(\ref{eq:at1}) is
\beq
I''=\sum_{n=0}^{N-1}\frac{1}{\cos \theta+ \lambda_k \omega^{-n}},
\eeq
with $\omega=e^{2 \pi i /N}$.
For even $k$, $\lambda_k= \omega^k$, and the sum can be re-written as $I''=$
\[
\begin{split}
= \sum_{n=0}^{N-1}\omega^{n} \sum_{l=0}^{\infty}(-\cos \theta\,\omega^{n})^l=\sum_{l=0}^{\infty} \cos^l \theta \sum_{n=0}^{N-1}\omega^{n(l+1)}.
\end{split}
\] 
Using $ \sum_{n=0}^{N-1}\omega^{n(l+1)}=N\, \delta[(l+1)\, \text{mod}\,N, 0]$ the sum becomes
\beq
I'' = N\sec \theta \sum_{l=1}^{\infty}(-\cos ^N \theta)^l=
\frac{N  \cos^{N-1} \theta}{1+\cos^N \theta}.
\eeq
For  $k$  odd  $\lambda_k= -\omega^k$, and following the same steps as above gives $I''=- N  \cos^{N-1} \theta/(1-\cos^N \theta)$ and hence (in all cases $N$ is odd),
\beq
I''=(-1)^k \frac{N \cos^{N-1} \theta}{1+(-1)^k \cos^N \theta}.
\eeq
\label{app:off diagonal}
\bibliography{pswpaper}
\end{document}